\documentclass[english,12pt]{article}
\usepackage{lmodern}

\usepackage[T1]{fontenc}
\usepackage[dvipsnames,svgnames,x11names,hyperref]{xcolor}
\usepackage[margin=0.6truein]{geometry}
\usepackage{babel}
\usepackage{booktabs}
\usepackage{amsmath}
\usepackage{amsthm}
\usepackage{amssymb}
\usepackage{stmaryrd}
\usepackage{graphicx}
\usepackage{setspace}
\usepackage{caption}
\usepackage[ruled,vlined]{algorithm2e}
\usepackage{subcaption}
\usepackage{comment}
\usepackage{natbib}
\usepackage{soul}
\usepackage[unicode=true,bookmarks=true,bookmarksnumbered=false,bookmarksopen=false,breaklinks=false,pdfborder={0 0 0},pdfborderstyle={},backref=page,colorlinks=true]{hyperref}
\hypersetup{linkcolor=RoyalBlue,citecolor=RoyalBlue, urlcolor = RoyalBlue}
\usepackage{colortbl}
\makeatletter

 \theoremstyle{remark}

 \theoremstyle{definition}
 
 \newtheorem{prop}{\protect Proposition}[section]

\usepackage{xr-hyper}
\usepackage{cleveref}
\newcommand{\argmax}{\operatornamewithlimits{argmax}}

\newcommand{\pr}{\mathrm{pr}}
\def\d{\mathrm{d}}
\newcommand{\btheta}{\boldsymbol{\theta}}
\newcommand{\bbeta}{\boldsymbol{\beta}}

\newcommand{\Yk}{\mathbf{Y}_k}
\newcommand{\Sk}{\mathbf{S}_k}
\newcommand{\Ck}{\mathbf{C}_k}

\newcommand{\Data}{\mathbf{D}}
\newcommand{\DataP}{\mathbf{D}^Y}

\makeatother

\providecommand{\examplename}{Example}

\newcommand{\secname}{Section}

\title{A cautious use of auxiliary outcomes for decision-making in randomized clinical trials.}
\date{}

\author{Massimiliano Russo\thanks{Department of Statistics, the Ohio State University, Columbus Ohio, U.S.A, russo.325@osu.edu},
Steffen Ventz\thanks{Division of Biostatistics and Health Data Science, School of Public Health, University of Minnesota, Minneapolis, MN, U.S.A.},
      and Lorenzo Trippa\thanks{T.H. Chan School of Public Health,  and Dana-Farber Cancer Institute, Boston, U.S.A.}}

\begin{document}
\maketitle
\begin{abstract}
\setstretch{1.0}
Clinical trials often collect data on multiple outcomes, such as overall survival (OS), progression-free survival (PFS), and response to treatment (RT). In most cases, however, study designs only use primary outcome data for interim and final decision-making. In several disease settings, clinically relevant outcomes, for example OS, become available years after patient enrollment. Moreover, the effects of experimental treatments on OS might be less pronounced compared to auxiliary outcomes such as RT. We develop a Bayesian decision-theoretic framework that uses both primary and auxiliary outcomes for interim and final decision-making. The framework allows investigators to control standard frequentist operating characteristics, such as the type I error rate and can be used with auxiliary outcomes from emerging technologies, such as circulating tumor assays. False positive rates and other frequentist operating characteristics are rigorously controlled without any assumption about the concordance between primary and auxiliary outcomes. We discuss algorithms to implement this decision-theoretic approach and show that incorporating auxiliary information into interim and final decision-making can lead to relevant efficiency gains according to established and interpretable metrics. 
\end{abstract}

\section{Introduction}
Clinical trials evaluate the efficacy and safety of experimental treatments. In most cases, regulatory decisions on the efficacy of treatments are based on comparisons of primary outcomes between the experimental treatment and the standard of care (SOC). In oncology, these comparisons typically attempt to demonstrate improvements in  overall survival (OS). However, in several disease settings, collecting data on OS requires long follow-up times. Therefore, using only OS for decision-making may be inefficient and expose patients to potentially ineffective experimental treatments.

Some outcomes are available earlier than the primary outcomes of interest and may provide stronger signals of treatment effects, potentially reducing trial size and duration~\citep[e.g.,][]{lothar:2012,chen:2019}. In oncology, outcomes such as progression-free survival (PFS), response to treatment (RT), and tumor size reduction have also been used in regulatory decisions. When these are used in place of the primary outcomes for decision-making, they are referred to as \textit{surrogate outcomes}. Alternatively, investigators may report summaries of these outcomes without using them  for decision-making,  during and or the end of the study. Also, in some trials, outcomes such as PFS inform only interim decisions (e.g., early stopping), without influencing the final reporting of primary results and recommendations at the end of the study.

An ideal surrogate outcome should unambiguously inform on the effects of experimental treatments on the primary outcome (see \cite{prentice:1989} for formal definitions). However, the relationship between treatment effects on primary and surrogate outcomes is often hypothetical or controversial. For instance, in glioblastoma, lomustine combined with bevacizumab showed improved PFS compared to lomustine alone, but later studies indicated no effect on OS~\citep{wick:2017}. 
There is extensive literature on how to evaluate agreement between the assessment of experimental therapies based on primary outcomes and evaluations based on surrogate endpoints \citep[e.g.,][]{buyse:2000,elliot:2014,vandenberghe:2018}. When prior evidence indicates strong agreement, investigators may consider using the surrogate as the primary endpoint in future trials, while weighing its benefits and risks.

In this paper, we examine the use of {\it auxiliary outcomes}---outcomes that have not been rigorously validated, such as tumor response in glioma---to support decision-making in clinical trials. When  trials are conducted, regulators would not accept these outcomes as surrogates and would require evidence of treatment effects on the primary outcome for drug approval. We develop a method for making interim and final clinical trial decisions---such as recommendations for treating future patients with the experimental therapy---that incorporates both primary and auxiliary outcomes while adhering to strict regulatory requirements regarding the risks of false positive results. Examples of regulatory requirements include the control  the family-wise error rate (FWER) if the treatment is tested in multiple subgroups.
We illustrate the potential advantages of joint modeling of auxiliary and primary outcomes discussing two  examples:

{\it 1-Trials designs with patient subgroups.} The aim is to provide evidence of treatment effects for specific subgroups while controlling the FWER~\citep[e.g.,][]{leblanc:2009,chugh:2009}. Auxiliary outcomes can help identify which subgroups benefit most from the experimental treatment, thereby increasing the likelihood of reporting positive results without compromising the  control of the FWER.

{\it 2-Multi-stage trial designs~\citep[e.g.,][]{simon:1989, ensign:1994}.} Auxiliary outcomes can support early futility decisions while ensuring power above a desired threshold and strict control of the probability of a false positive result. Unpromising results on auxiliary outcomes can lead to early termination of the study, potentially reducing resources spent on ineffective treatments.

In both cases, the regulatory constraints are defined through standard frequentist requirements. 

We rely on the Bayesian decision-theoretic framework~\citep{berger:1985} to incorporate primary and auxiliary outcomes for decision-making. Throughout the paper, the optimal clinical trial design, which includes a detailed plan for interim and final decisions,  coincides with the solution of a constrained maximization problem 
defined by three key components:  (i) a utility function \citep{lindley:1976}; (ii) a Bayesian prior model representative of previous studies that includes primary and auxiliary outcomes;  and (iii) constraints on frequentist operating characteristics, such as the rigorous control of the probability of a false positive result below a pre-specified $\alpha$ level, enforced across all potential scenarios in which the experimental treatment does not improve the primary outcomes. Importantly, the design is optimized using a Bayesian model, but also satisfies frequentist constraints. 

Our decision-making approach considers the perspectives of two distinct stakeholders. The utility function and prior model represent the viewpoint of the investigators or the pharmaceutical company. Their primary objective is to select a study design that maximizes expected utility (e.g., projected revenue if the drug succeeds) while minimizing losses (e.g., sunk costs if development is halted), given the information incorporated in the prior model. From a Bayesian decision-theoretic perspective, the optimal design for this stakeholder maximizes the expected utility. However, using this ideal, unconstrained optimization is often impractical due to regulatory requirements imposed by a different stakeholder, the regulator (e.g., the FDA). Consequently, our optimal designs will aim to maximize expected utility criteria
within a constrained subset of candidate study designs that comply with regulatory requirements (e.g., family-wise error rate or statistical power). 

Differences in how investigators and regulators view auxiliary outcomes are often evident. Investigators typically expect auxiliary outcomes to be representative of the clinical benefit captured by the primary outcomes. At the same time, regulators require direct evidence of treatment effects on the primary outcomes, with stringent control of false positives and bias. For example, in breast cancer research, trial sponsors have long regarded early responses to therapy as promising predictors of long-term benefit. Over the past two decades, however, regulatory agencies have progressively changed their recommendations on the use of early outcomes, informed by large-scale evaluations of candidate surrogate endpoints across multiple clinical trials~\citep{spring:2020}.

Bayesian decision theory has been used to design clinical trials~\citep{lewis:1994,ventz:2015,thall:2019}. It is useful for selecting a study design or a method, such as a testing procedure, among candidate designs/methods through an interpretable measure of utility (or loss) of the trial, tailored to the study aims \citep{arfe:2020}. To the best of our knowledge,  earlier work has not investigated Bayesian designs that include auxiliary outcomes and attain a frequentist control of false positive results across all scenarios without positive effects on the primary outcomes. We show that the joint use of primary and auxiliary outcomes for interim and final decisions can increase the utility of the clinical trial without compromising the rigorous control of frequentist operating characteristics, such as the FWER.

Controlling for frequentist characteristics is a key aspect of our approach compared to other Bayesian designs, which often rely on simulations to evaluate these characteristics before the onset of the trial. Selecting representative scenarios for these simulations can be challenging~\citep{han:2024}. Moreover, a limited number of scenarios without treatment effects on primary outcomes may not adequately capture the variation in false-positive rates across different joint distributions of covariates and outcomes. Other Bayesian approaches~\citep{calderazzo:2020} advocate relaxing frequentist constraints so that they hold on “average” rather than pointwise. By combining Bayesian models with utility functions, these methods relax frequentist criteria in a valuable way and warrant further development and discussion.

The paper proceeds as follows: \secname{s}~\ref{sec:notation}~and~\ref{sec:betterDecisions} introduce the notation and explain how we leverage auxiliary outcomes in  decisions during and at the end of clinical trials. \secname~\ref{sec:approximation} describes an approximation strategy that can be used to design clinical trials.
\secname{s}~\ref{sec:weighted_bonferroni} and~\ref{sec:sequential} discuss the operating characteristics of designs that integrate primary and auxiliary outcomes for decision making.  Finally, \secname~\ref{sec:discussion} provides a brief discussion. Code for the computational procedures used to approximately identify optimal trial designs that satisfy the specified frequentist operating characteristics  is available at \url{https://github.com/rMassimiliano/primary\_and\_auxiliary}.

\section{Primary and auxiliary outcomes} \label{sec:notation}
\subsection{Notation} 
We consider a randomized controlled trial (RCT) enrolling up to $N$ patients, comparing an experimental treatment with SOC in $K\ge1$ pre-defined subgroups (e.g., by demographics or biomarkers). Also, $N_k\ge0$ indicates the number of enrollments in subgroup $k=1,\ldots,K$ at the end of the trial.
 
{\it Data and prior model.}
For each patient $i= 1, \ldots, N_k$ in subgroup $k=1, \ldots, K$,  $Y_{k,i} \in \mathcal Y$  indicates the primary outcome (e.g., OS),  $S_{k,i} \in \mathcal S$ the auxiliary outcome (e.g., RT), and $C_{k,i} \in\{0,1\}$  the assignment to the SOC ($C_{k,i} =0$) or experimental treatment ($C_{k,i}=1$). 
In what  follows $\Data_k = \{Y_{k,i},S_{k,i},C_{k,i}\}_{i=1}^{N_k} \buildrel iid \over \sim p_{\theta_k}(\Data_k)$ indicate the data ($\Data_k$) and the distributions ($ p_{\theta_k}$) for subgroups $k=1, \ldots, K$, with unknown parameters $\btheta = (\theta_1, \ldots, \theta_K) \in \Theta$. We use $p_{\btheta}(\Data)$ for the joint distribution of  $\Data = \{\Data_k\}_{k=1}^K$,  including patients from various subgroups, and $\mathcal D$ to refer to the sample space. We use $\DataP_k = \{Y_{k,i},C_{k,i}\}_{i=1}^{N_k}$, $\DataP = \{\DataP_k\}_{k=1}^K$ to indicate the primary outcome data. Information on $\btheta$ is expressed through a prior distribution $\pi(\btheta)$ on $\Theta$ that incorporates the prior belief of the investigator or pharmaceutical company. 

{\it Treatment effects on the primary outcomes.}
Treatment effects on the primary outcome in group $k$ is quantified by a summary $\gamma_k=\gamma_k(\btheta)$, where $\gamma_k >0$ indicates a positive treatment effect in group $k$.  For example, $\gamma_k = \mathbb E[Y_{k,i} \mid C_{k,i} = 1] - \mathbb E[Y_{k,i} \mid C_{k,i} = 0]$ indicates the subgroup-specific difference of the average primary outcome under the experimental treatment and the SOC. We will test subgroup-specific null hypotheses $H_{0,k}: \gamma_k \leq 0$ for groups $k = 1,\ldots, K$.

{\it Sequential decisions.}
In RCT, data are often analyzed sequentially at interim analyses (IAs) in $T$ stages, potentially leading to major decisions such as early termination of the study or an increase of the sample size $N$. When necessary (\secname~\ref{sec:sequential}), we use the superscript $t$ to indicate data up to stage $t$. In particular, $\Data^{t}$ will indicate data up to the $t$-th stage of the trial,  $t=1,\ldots, T$, and $\Data = \Data^{T}.$ 

\subsection{Optimization of the expected utility with regulatory constraints}

{\it Actions.}
During the trial, $M$ binary decisions (or actions) $\mathbf{a} = (a_1, \ldots, a_M) \in \mathcal A= \{0,1\}^M$ are made. For example, when $K=1$ and $T>1$, the $M=T$ decisions $\{a_m\}_{m=1}^T$, might refer to continuing ($a_m = 1$) or stopping ($a_m = 0$) the study early for futility at IA $m=1,\ldots, T-1$, and reporting evidence of treatment effects ($a_T = 1$) or not ($a_T = 0$) at the end of the trial.
 
{\it Decision functions.}
Decisions during the trial are made using a function $\varphi : {\tt data} \to \mathcal A$  which maps available data into actions $\mathbf{a}$. For example, futility IAs and the final analysis (FA) at times $t=1, \cdots, T$ use the available data up to stage $t$. In this case, the decisions  are $\varphi(\Data)= (\varphi_{1}(\Data^1), \varphi_{2}(\Data^2), \ldots, \varphi_{T}(\Data^T))$. We let $\varPhi$ indicate the space of all potential decision functions that can be used during the trial. 

{\it Regulatory constraints.}
The decision function has to meet regulatory requirements, such as controlling the FWER. Regulatory constraints are formalized by the inequality $\mathbb E_{\btheta}[h(\varphi(\Data), \btheta)] \leq \alpha$, for every $\btheta \in \Theta.$ The expectation is with respect to the unknown distribution of the data $p_{\btheta}(\Data)$.  For example, when $M=K=1$ and $\varphi(\Data)$ is the result of the final testing  of treatment effects, the function $h(\varphi(\Data),\btheta) = \varphi(\Data) \times 1\{\gamma_1 \leq 0 \}$ indicates if there is a false positive result or not. Therefore, if $\mathbb E_{\btheta}[h(\varphi(\Data), \btheta)] \leq \alpha$ for every $\btheta \in \Theta$, then $\varphi$ controls the type I error rate at level $\alpha$. Similarly, when testing $\{H_{0,k}\}_{k=1}^K$ for $K>1$ subgroups, with $T=1$ and $M=K$,  the function $h(\varphi(\Data), \btheta ) = 1 \{ \sum_{k=1}^K 1\{ \gamma_k \leq 0\} \varphi_{k}(\Data) > 0  \}$ can be used to control the FWER. The expectation $\mathbb E_{\btheta}[h(\varphi(\Data), \btheta ) ]$ is equal to the probability of reporting at least one false positive result $\mbox{pr}_\theta \left(\cup_{k: \gamma_k \leq 0}  \{\varphi_{k}( \mathbf{D}) =1\}\right)$. For a specific function $h$,  $\Phi_{\alpha,h} \subset \varPhi$ indicates the set of decision functions $\varphi$ that satisfy the regulatory constraints $\mathbb E_{\btheta}[h(\varphi(\Data), \btheta)] \leq \alpha$ for all $\btheta \in \Theta$. We use $\Phi'_{\alpha,h} \subset \varPhi$ to indicate the subset of decision functions that satisfy the regulatory constraints and ignore the auxiliary outcome, i.e., those that use only the available information on the primary outcomes $\Data^Y$.

{\it Utility functions.} 
In what follows $u(\mathbf{a}, \btheta)$ indicates a {\em utility function}, which depends on the unknown $\btheta$ parameters and the actions $\mathbf{a}$. Appropriate utility functions are crucial for decision-theoretic procedures. In clinical trials, the utility function typically summarizes the study's costs and the potential benefits of demonstrating the efficacy of an experimental therapy, representing relative preferences and making explicit the goals of the study. We refer to \citet{thall:2019}  and \citet{lee:2022} for examples and discussions on the elicitation of utility functions in cancer research. These authors, among others,  developed utility functions that capture multiple aspects of a clinical study, such as drug-related adverse events and patient quality of life. When $M=K$, $T=1$, and the action $\mathbf{a}$ coincides with rejecting or not the null hypotheses $\{H_{0,k}\}_{k=1}^K$,  then an interpretable utility function $u(\mathbf{a},\btheta) = \sum_{k=1}^K [1\{a_k=1, \gamma_k> 0\} -\lambda 1\{a_k=1, \gamma_k\leq 0\}]$ includes a unitary reward for each true positive result and a penalty $(\lambda >0)$ for each false positive result. This utility function can represent economic interests, including potential gains and costs associated with true and false positive subgroup-specific results. Importantly, other utility functions can represent different interests, for example the life expectancy of future patients diagnosed with a specific disease; see \citet{berry1995adaptive} for a discussion. 

{\it Translating primary and auxiliary outcomes into actions.} 
In the Bayesian decision-theoretic framework~\citep{berger:1985}, without regulatory restrictions, the decision maker selects the decision function $\varphi^* \in \varPhi$ that maximizes the expected utility $\mathbb E^{\pi} [u(\varphi(\Data),\btheta)]= \int_{\mathcal D \times \Theta} u(\varphi(\Data),\btheta) d\pi(\btheta, \Data)$, where the expectation is computed with respect to the joint distribution $\pi(\btheta, \Data) = \pi(\btheta) p_{\btheta}(\Data)$ of the data $\Data$ and the parameters $\btheta$. Note that $\varphi^*$ translates  the available primary $(Y_{k,i})$ and auxiliary $(S_{k,i})$ data into actions $\varphi^*(\Data).$

{\it The decision-theoretic paradigm with regulatory constraints.} 
The decision function $\varphi^*(\Data)$  is selected without accounting for regulatory constraints. In contrast, we propose to select the decision function $\varphi_{Y,S}$ that maximizes the expected utility among those that are compliant with the regulatory constraints, $\mathbb E_{\btheta}[h(\varphi(\Data), \btheta)] \leq \alpha$ for all $\btheta \in \Theta$. That is, 
\begin{equation}
	\varphi_{Y,S} \in \argmax_{\varphi \in \varPhi_{\alpha,h} } \mathbb E^\pi[u(\varphi(\Data),\btheta)]. 
 \label{eq:maximization}
\end{equation}
The constraint $\varphi \in \varPhi_{\alpha,h}$ can be used to consider only testing procedures that control the type I error rate at a desired $\alpha$-level across all scenarios without positive effects on the primary outcomes.

The maximization in expression \eqref{eq:maximization} in most cases is analytically intractable. In Section~\ref{sec:approximation} we will consider approximations of this optimization problem.
\begin{table}[h!]
	\center
\resizebox{\textwidth}{!}{
\begin{tabular}{ll}
	\toprule
	$Y_{k,i}\in \mathcal Y$ & primary outcome, subject $i$ in group $k$ \\
	$S_{k,i}\in \mathcal S$ & auxiliary outcome, subject $i$ in group $k$\\
	$C_{k,i} \in \{0,1\}$ & treatment assignment, subject $i$ in group $k$\\
 $\Data = \{\Yk,\Sk,\Ck\}_{k=1}^K \in \mathcal D$ & trial data \\ 
 $\DataP = \{\Yk,\Ck\}_{k=1}^K\in \mathcal D^Y$ & trial data excluding auxiliary outcomes \\ 
	$p_{\btheta}(\Data)$ & distribution of the clinical trial dataset \\
	$\btheta =(\theta_1, \ldots, \theta_K) \in \Theta$ & group-specific parameters \\
	$\pi(\btheta)$ & prior distribution\\
	$\pi(\Data,\btheta) = p_{\btheta}(\Data)\pi(\btheta)$ & prior model \\
	$\gamma_k=\gamma_k(\btheta)$ & treatment effects on the primary outcome in group $k$\\
	 $\mathbf{a} = (a_1, \ldots, a_M) \in \mathcal A= \{0,1\}^M$ & a vector of actions \\
$\varphi(\cdot) = [\varphi_1(\cdot), \ldots, \varphi_M(\cdot)] $ & decision function that maps data into actions\\
$\varPhi$ & space of decision functions\\
$\varPhi_{\alpha,h}$ & space of regulator-compliant decision functions; $\mathbb E_{\btheta}[h(\varphi(\Data), \btheta)] \leq \alpha$ \\
$\varPhi'_{\alpha,h}$ & space of regulator-compliant decision functions that do not use information on the auxiliary outcomes \\
$\varphi^* $ & optimal choice within $\varPhi$, ignoring the regulatory constraints\\
$\varphi_{Y,S} $ & optimal choice within $\varPhi_{\alpha,h}$, compliant with the regulatory constraints\\
$\varphi_{Y} $ & optimal choice within $\varPhi'_{\alpha,h}$, compliant with the regulatory constraints\\
$u(\mathbf{a}, \btheta)$ & utility expressed as a function of actions $\mathbf{a}$ and unknown parameters $\btheta$ \\
$\mathbb E_\theta[\cdot]$ & expectation  with respect to the unknown distribution of the data $p_{\btheta}(\Data)$\\
$\mathbb E^\pi[\cdot]$ & expectation with respect to the joint distribution $\pi(\btheta, \Data) = \pi(\btheta) p_{\btheta}(\Data)$ of the data $\Data$ and parameters $\btheta$.\\
\bottomrule
\end{tabular}
}
\caption{Notation used throughout the paper.}
\end{table}

\section{Auxiliary outcomes can improve decisions}\label{sec:betterDecisions}
\label{test}
This section illustrates how incorporating auxiliary outcomes can improve decision-making while maintaining strict control of regulatory constraints. The extent of these improvements depends on the specified prior and utility criteria underlying the trial design. Using a stylized example, we compare the utility of an optimal decision function $\varphi_{Y,S} \in \varPhi_{\alpha,h}$, which uses both auxiliary and primary outcomes, to an optimal decision $\varphi_{Y} \in \varPhi'_{\alpha,h}$ that relies solely on primary outcomes. We focus on a simplified setting without subgroups and a single hypothesis test at the trial’s conclusion ($M=K=1$), omitting indices $m$ and $k$ and denoting parameters as $\theta$ and $\gamma$. While this paper focuses on randomized studies, for clarity, we consider a single-arm design with a sample size of one, assuming the probability of a positive primary outcome under control is known. In this example, $\varPhi_{\alpha,h}$ and $\varPhi'_{\alpha,h}$ represent sets of testing procedures controlling type I error at level $\alpha$; that is, $h(a,\btheta) = a \times 1\{\gamma(\theta) \leq 0\}, a=0,1$. If the null hypothesis $H_0: \gamma(\theta) \leq 0$ holds, then $p_\theta(\varphi(\Data)=1)<\alpha$ for every $\varphi \in \varPhi_{\alpha,h}$. Functions in $\varPhi_{\alpha,h}$ use both primary and auxiliary outcomes, whereas those in $\varPhi'_{\alpha,h} \subset \varPhi_{\alpha,h}$ use only primary outcomes, all under equivalent regulatory constraints.

\noindent \textbf{Example 3.1} {\em We consider a single-arm trial with binary primary and auxiliary outcomes $(Y_i,S_i) \in \{0,1\}^2$. The outcomes $Y_i$ and $S_i$ are independent, with $p_\theta(y,s) = {\theta^y_Y} (1-\theta_Y)^{1-y} \theta^s_S (1-\theta_S)^{1-s}$ and $\theta=(\theta_S,\theta_Y)$. The probability of a positive primary outcome under the control therapy is known and equal to $0.05.$ The study tests the null hypotheses $H_0: \theta_Y \leq 0.05$ versus $H_1: \theta_Y > 0.05$, $\alpha = 0.05$ and the utility function is 
\begin{align}\label{Exa:Ut:1}
u(a, \btheta) = \begin{cases}
1 & \mbox{ if } a = 1 \mbox{ and } \theta_Y > 0.05, \\
		-\lambda & \mbox{ if } a = 1 \mbox{ and } \theta_Y \leq 0.05,\\
 0 & \mbox{ otherwise.}
 \end{cases}
\end{align}
Utility function~\eqref{Exa:Ut:1} gives reward $1$ if $a=1$ (reject $H_0$) and $\theta_Y>0.05$,  loss $\lambda=100$ if $a=1$ and $\theta_Y\le0.05$, and is equal to $0$ otherwise.

We assume that prior studies suggest concordance between auxiliary and primary outcomes, expressed through the prior distribution $\pi(\theta) = \pi(\theta_Y) \pi(\theta_S\mid \theta_Y)$, with $\theta_Y \sim Unif(0,1)$ combined with a deterministic relation between $\theta_S$ and $\theta_Y$.  In particular, $\theta_S \mid \theta_Y \sim \delta_0(\theta_S)$ for every value of $\theta_Y \leq 0.05$, and symmetrically $\theta_S \mid \theta_Y \sim \delta_1(\theta_S)$  for every $\theta_Y > 0.05$. Here $\delta_x(\cdot)$ is a probability distribution that assigns probability $1$ to a single point $x$ and $0$ everywhere else. 

For simplicity, $N=1$ and the sets of candidate decision functions $\varPhi_{\alpha,h}$ with full information $(Y,S)$ and $\varPhi'_{\alpha,h}$ with partial information $(Y)$ include only non-randomized decisions (i.e., map the data into $\{0,1\}$). 
In this example, we can enumerate the functions in $\varPhi$. Since there are only four possible configurations of the study dataset $(Y,S)=(0,0)$, $(0,1)$, $(1,0)$, or $(1,1)$, the full set $\varPhi$ of candidate decision functions has cardinality $2^4$. Only $4$ functions in $\varPhi$ satisfy the regulatory constraint. In particular, there are two functions in $\varPhi'_{\alpha,h}$ that control the type I error rate and do not use auxiliary information: \\
(i) $\varphi$ never rejects $H_0$, i.e. $\varphi(\Data) = 0$, with expected utility 0, and \\
(ii) $\varphi$ rejects $H_0$ when $Y=1$, i.e. $\varphi(\Data) = Y$, with expected utility $\approx 0.37$. \\ There are two other candidate decision functions in $\varPhi_{\alpha,h}$ that control the type I error rate and leverage auxiliary information: \\
(iii) $\varphi$ rejects $H_0$ when $Y=S=1$, i.e. $\varphi(\Data) = Y \times S$, with expected utility $\approx 0.5$, and\\ 
(iv) $\varphi$ rejects $H_0$ when $Y=1-S=1$, i.e. $\varphi(\Data) = Y(1-S)$, with expected utility $\approx -0.13$.\\ 
Although counterintuitive, the decision function (iv) is a legitimate $\alpha$-level test. 
Note that the optimal decision function $\varphi_{Y,S}$ with full information $(Y,S)$ 
has higher expected utility ($\approx 0.5$) than the optimal decision function $\varphi_Y$ with partial information ($\approx 0.37$).
}

Our stylized example considers an informative and restrictive prior $\pi$ on $(\theta_Y,\theta_S)$ that can be easily modified to obtain a  different prior with full support (i.e., $[0,1]^2$). The example shows that with $\lambda=100$,  the inclusion of the auxiliary outcome for decision-making leads to the selection of the decision function $\varphi_{Y,S}= Y\times S$, which has the highest expected utility among all functions in $\varPhi_{\alpha,h}$. With other strictly positive values of $\lambda$, we would select the same decision function $\varphi_{Y,S}$. However, with $\lambda=0$,  the optimal decision function subject to regulatory constraints, $\varphi_{Y,S}$ (with auxiliary outcome) and $\varphi_{Y}$ (without auxiliary outcome) have identical expected utility ($\approx 0.5$).

This equality with $\lambda =0$ suggests that when we consider an $\alpha$-level test that maximizes the Bayesian expected power (BEP), a popular utility function,  that is 
 \begin{equation}
\varphi_{Y,S} \in \argmax _{\varphi \in \varPhi_{\alpha,h} }
 \mathbb E^\pi[ \varphi(\Data) 
 ],
\label{eq:max_1}
\end{equation} 
the use of auxiliary outcomes does not lead to an increase in the expected utility compared to the optimal decision functions based only on primary outcomes 
(Proposition~\ref{th:no_free_lunch}). Note that the solution of \eqref{eq:max_1} could be a randomized test. In this case, the decision function $\varphi$ would map the data into the interval $[0,1]$~\citep{lehmann:2006}. Randomized tests are rarely used in practice, but they can be helpful to compare  $\varphi_{Y,S}$ and $\varphi_Y$. In Proposition~\ref{th:no_free_lunch} below, we consider randomized tests, and again $K=M=1$. The proposition considers two sets of random variables ($S$ and $Y$). The class of possible $(S,Y)$ distributions is indexed by $\Theta$, 
and the null hypothesis $H_0$ is a subset of  $Y$ marginal distributions, i.e., a subset of $\Theta$. For simplicity, we consider a single-arm study. 
\begin{prop}\label{th:no_free_lunch}
The set $\argmax_{\varphi \in \varPhi_{\alpha,h} }
 \mathbb E^\pi[ \varphi(\Data) ]$
 contains at least one decision function in $\varPhi'_{\alpha,h}$ 
 \end{prop}
The proposition states that there exists at least one $\alpha$-level test $\varphi\in \varPhi_{\alpha,h}$ that maximizes the BEP, expression \eqref{eq:max_1}, without leveraging auxiliary information $S$ (i.e., $\varphi \in \varPhi'_{\alpha,h}$).  The proof is provided in the Supplementary Materials. 
Proposition~\ref{th:no_free_lunch} indicates that for some combinations of prior model $\pi$, regulatory constraints $h(a,\btheta)$, and utility function $u$, the auxiliary data can not improve the expected utility.   But for other combinations of  $\pi$,  constraints $h(a,\btheta)$, and utility criteria $u$,  the improvements are substantial.   In particular, in \secname{s}~\ref{sec:weighted_bonferroni} and~\ref{sec:sequential}, we will provide realistic examples in which including auxiliary data increases expected utility and results in substantial improvements in relevant operating characteristics.

\section{Approximation of the optimal solution}\label{sec:approximation}
Solving the constraint maximization problem~\eqref{eq:maximization} is often unfeasible as it requires a maximization over the space of functions $\varPhi_{h,\alpha}$, which only includes candidate functions $\varphi \in \varPhi$ that satisfy the inequality $\mathbb E_{\btheta}[h(\varphi(\Data), \btheta)] \leq \alpha$ for all $\boldsymbol{\theta} \in \Theta.$   For practical implementations, we propose to approximate the maximization problem~\eqref{eq:maximization} by focusing on a convenient subset of functions $\varPhi_{\mathcal B} \subset \varPhi$.
 
{\it Data summaries $(W)$.} 
A first step to define  $\varPhi_{\mathcal B}$ is the inclusion in this subset of  decision functions $\varphi : {\tt data} \rightarrow \mathcal{A}$ that map a vector of data summaries $W(\Data) \in \mathbb R^v$, such as estimates of the treatment effects on primary and auxiliary outcomes, into decisions, $\mathbf{a}\in\mathcal{A}$. If $W(\Data)$ is a sufficient statistic for $p_{\btheta}$, there is no loss of information, while informative data summaries $W(\Data)$ can be the basis for nearly optimal decisions in $\varPhi_{\alpha,h} $with minimal reductions of the expected utility. Low-dimensional summaries also simplify the interpretation of candidate decision functions. If $W(\Data)$ includes treatment effect estimates, say the difference of the mean primary and auxiliary outcomes in the experimental and SOC arms, then each function in $\varphi \in \varPhi_{\mathcal B}$  translates these estimates into decisions. In the next sections, we further discuss the interpretability of decision functions through examples. The vector $W(\Data)$ might include summary statistics that are commonly used for data analysis, such as Z-statistics or p-values (see Section~\ref{sec:weighted_bonferroni}), but also Bayesian summaries of a posterior distribution. For example, the posterior probability of positive treatment effects  or the predicted probability that a trial will report positive findings can be used (see Section~\ref{sec:sequential}).

{\it Parametric approximations ($\varPhi_{\mathcal B} \subset \varPhi$).} 
We further simplify the problem and restrict our selection of decision functions to parametric maps $\varphi_\beta : {\tt data} \rightarrow \mathcal{A}$ mediated by $W(\Data)$. The candidate decision functions become $\varPhi_{\mathcal B} = \{ \varphi_{\bbeta} (\cdot); \bbeta \in \mathbb R^q \}$. These functions transform $W(\Data)$ into decisions and are indexed by  $\bbeta \in \mathbb R^q$.  The subset $\varPhi_{\mathcal B}\subset \varPhi$ can achieve a sufficient degree of flexibility---a modest loss of utility compared to the optimization in \eqref{eq:maximization}---through (i) an appropriate choice of the family of approximating functions and (ii) the number of parameters $q$. For example, for a single decision $M=K=1$, and a single $W(\Data) \in \mathbb R$ data summary, we can use the sign of a $q$-degree polynomial to make a decision, say testing a generic $H_0$. In this example, polynomials can approximate a generic  map $W(\Data) \rightarrow \{-1,+1\}$. 

{\it Regulatory constraints on parametric decision functions.} 
Once $W(\Data)$ and $\varPhi_{\mathcal B}$ have been specified, we define $\mathcal B_{\alpha, h} = \{ \bbeta \in \mathbb R^q : \mathbb E_{\btheta}[h(\varphi_{\bbeta} (\Data), \btheta)] \leq \alpha, ~ \forall \btheta \in \Theta \}$ as the subset of parametric functions that satisfy the regulatory constraints. Next, we select $ \bbeta_{ Y, S} $ in $ \mathcal B_{\alpha, h}$ that maximizes the expected utility, 
\begin{equation}
	 \bbeta_{ Y, S} \in \argmax_{\bbeta \in \mathcal{B}_{\alpha,h} } \mathbb E^\pi[u(\varphi_{\bbeta}(\Data),\btheta)],
 \label{eq:par_maximization}
\end{equation}
and use $ \varphi_{ \bbeta_{ Y, S} } (\cdot)$ as approximate solution of \eqref{eq:maximization}. 

To evaluate if a candidate $\bbeta $ satisfies regulatory constraints, in the following sections, we use large sample approximations of the expectations
$\mathbb E_{\btheta}[h(\varphi_{\bbeta} (\Data), \btheta)]$. Leveraging results from the statistical literature for some choices of $W(\Data)$ and $\varPhi_{\mathcal B}$ is particularly convenient. In fact, asymptotic results can reduce the computational effort required to select an (approximately) optimal design that the regulator can accept. We also discuss a bootstrap approach to estimate  $\mathbb E_{\btheta}[h(\varphi_{\bbeta} (\Data), \btheta)]$, and ultimately to comply with the regulator requests, especially in settings with a small sample size (see Section~\ref{sec:boot_calibration}).

In the next two sections, we consider two common decision problems in clinical trials, discuss implementations of the outlined approach, and evaluate the operating characteristics produced by the approximate solution in expression \eqref{eq:par_maximization}.

\section{The use of auxiliary outcomes to test treatment effects in subgroups}\label{sec:weighted_bonferroni}

{\it Multiple hypothesis testing. } 
We consider a RCT with $K>1$ pre-specified non-overlapping biomarker subgroups. In each subgroup $k=1, \ldots, K,$ the aim is to test if the experimental treatment has positive effects on the primary outcomes, controlling the FWER at level $\alpha$. In this section $\varphi(\Data)=(\varphi_1(\Data),\ldots,\varphi_K(\Data))$ and $M=K$. Recall that the FWER is the probability of reporting  one or more false positive results, and can be expressed as $\mbox{pr}_\theta \left(\cup_{k: \gamma_k \leq 0}  \{\varphi_{k}( \mathbf{D}) =1\}\right) = \mathbb E_{\btheta}[h(\varphi(\Data), \btheta ) ]$ with the  function $h(\varphi(\Data), \btheta ) = 1 \{ \sum_{k=1}^K 1\{ \gamma_k \leq 0\} \varphi_{k}(\Data) > 0  \}.$  The FWER control  at level $\alpha$ across all $\btheta \in \Theta$  is a standard regulatory constraint~\citep[e.g.,][]{dmitrienko:2010}.

{\it Utility function.}
An interpretable utility function $u$ for the  testing problem that we described includes a unitary reward for each subgroup in which  positive treatment effects are correctly identified, and a penalty $\lambda \geq 0$ for each false positive result, 
\begin{equation}
u(\varphi(\Data),\btheta) = \sum_{k=1}^K \varphi_k(\Data)1\{ \gamma_k > 0 \} -\lambda
\sum_{k=1}^K \varphi_k(\Data)1\{ \gamma_k \leq 0 \}.
\label{eq:utility_ex1}
\end{equation}

{\it Hypothesis testing procedures.} 
Several methods  for testing multiple hypotheses $H_{0,k},$ $k=1,\ldots,K$, in clinical trials  are based on  subgroup-specific p-values $\text{\tt pv}_k = \text{\tt pv} (\DataP_k) $ ~\citep[e.g.,][]{dickhaus:2014}. Typically these p-values are exact or approximate tail probabilities $ \text{\tt pv}_k = p( U_{y,k} \geq u_{y,k} \mid H_{0,k})$. The statistic $U_{y,k}$ is selected to provide subgroup-specific evidence against $H_{0,k}$, and $u_{y,k}$ is its value computed at the end of the trial.

Weighted Bonferroni procedures cover a broad and widely used class of methods for hypothesis testing for controlling the FWER~\citep[e.g.,][]{roeder:2009, dobriban:2015}. A weighted Bonferroni procedure is specified through a set of non-negative weights $\boldsymbol{\omega} = (\omega_1,\ldots,\omega_K)$, one for each hypothesis $H_{0,1}, \cdots, H_{0,K}$, with $\sum_{k=1}^K \omega_k =1$, and decisions $\varphi_k(\DataP) = 1\{\text{\tt pv} (\DataP_k) \leq \omega_k \alpha \}$. In other words, a weighted Bonferroni procedure rejects $H_{0,k}$ when $\text{\tt pv}_k \leq \omega_k \alpha$. This approach coincides with the standard Bonferroni method when $\omega_k =1/K$ for $k=1,\ldots,K$. If the weights $\boldsymbol{\omega}$ are selected using information available before the onset of the RCT, and correlate with available evidence in favor $H_{0,k}$, then they can increase the expected number of true positive results compared to the standard Bonferroni approach, while still controlling the FWER at same level $\alpha$~\citep[][]{roeder:2009}. Some authors proposed data-driven weights~\citep[e.g.,][]{schuster:2004,kropf:2004,finos:2007}. 

In the remainder of this section, we discuss a modification of the weighted Bonferroni procedure based on the decision-theoretic framework outlined in \secname{s}~\ref{sec:notation} and~\ref{sec:approximation}.  In particular, we introduce   parametric  functions that  transform auxiliary data summaries into weights $\boldsymbol{\omega}$. The optimal weights are selected maximizing the expectation of the utility function~\eqref{eq:utility_ex1} with the constraint of a   FWER  below the desired $\alpha$ level.

To quantify the effects on the auxiliary outcome, we use the mean difference $\overline{S}_k = \overline{S}_{k,1} - \overline{S}_{k,0}$, where $\overline{S}_{k,c}$ indicates the subgroup-specific mean of the auxiliary outcome in arm $c.$  Large values of $\overline{S}_k$ indicate relevant effects on the auxiliary outcome. Here, we focus on contexts where investigators expect a positive association between effects on auxiliary and primary outcomes.  The data summaries $W(\Data) = (\{\text{\tt pv}_k\}, \{\overline{S}_k\})$, are convenient test statistics and treatment effect estimates commonly used in clinical trials. As described in Section~\ref{sec:approximation}, we specify a family of parametric decision functions $\varPhi_{\mathcal B}: \Data\rightarrow \mathcal{A}$ that depend on the data $\Data$ through the summaries $W(\Data) = (\{\text{\tt pv}_k\}, \{\overline{S}_k\})$, 
\begin{equation}\label{eq:procedure1}
	\varphi_{\bbeta}(\Data) = \Big(\varphi_{\bbeta,1}(\Data), \ldots, \varphi_{\bbeta,K}(\Data)\Big). \end{equation}
 In particular, $\varphi_{\bbeta,k}(\Data) = 1\{ \text{\tt pv}_k \leq \alpha \omega_k(\bbeta, \overline{S}\}) \}$, for $k=1, \ldots K$, with  $\overline{S} = (\overline{S}_1, \ldots, \overline{S}_K)$ and $ \bbeta \in \mathbb R^K.$ We define the weights $\boldsymbol{\omega} $ using the auxiliary outcomes, 
\begin{equation}
\omega_k(\bbeta, \overline{S}) = \frac{\exp\{ \beta_k \overline{S}_{k}\}}{\sum_{\ell=1}^K \exp\{\beta_\ell \overline{S}_\ell \}}, \mbox{ for } k=1,\ldots,K.
\label{eq:procedure1w}
\end{equation}
The decision function $\varphi_{\bbeta,k}(\Data)$ translates both primary and auxiliary outcomes into decisions. Therefore, two hypothetical replicates of the same trial with the same p-values ${\tt pv}_k, \; k=1,\ldots,K$, and differences in the estimates of the auxiliary treatment effects  $\overline{S}_k, \; k=1,\ldots,K$, can lead to different decisions.

To compute the p-values, we use subgroup-specific estimates of the treatment effects on the primary outcomes $\overline{Y}_{k}= \overline{Y}_{k,1} - \overline{Y}_{k,0}$, where $\overline{Y}_{k,c}$ indicates the subgroup-specific mean of the primary outcome in arm $c$, and use $\text{\tt pv}_k = 1- \Phi( Z_k )$, where $Z_k =\overline{Y}_{k}/\{\widehat{\mbox{Var}}(\overline{Y}_{k})^{1/2}\}$. Here $\widehat{\mbox{Var}}(\overline{Y}_{k})$ is a consistent estimate of the variance of $\overline{Y}_{k}$, and $\Phi(\cdot)$ is the cumulative distribution function of a standard normal random variable. Under ${H_{0,k}},$ the random variables $Z_{k}$ and $\text{\tt pv}_k$ converge to standard normal and uniform distributions as $N_k$ increases. The set $\varPhi_{\mathcal B}$ can be further restricted to parameter configurations with $\beta_1=\beta_2=\ldots=\beta_K$. We will consider this restricted version of $\varPhi_{\mathcal B}$ later in our simulation study.

Proposition~\ref{prop:fwer_asymp} illustrates that the outlined weighted Bonferroni procedure with data-dependent weights $\boldsymbol{\omega}$ defined in~\eqref{eq:procedure1w} controls the FWER asymptotically at level $\alpha$.

\begin{prop}	\label{prop:fwer_asymp}
Assume that for each $\btheta \in \Theta$ the means $\mathbb E_\theta( \vert S_{k,i} \vert \mid C_{k,i}=c) $ and variances $\mbox{Var}_\theta(Y_{k,i} \mid C_{k,i}=c) $ are finite for all $k=1, \ldots, K$ and $c=0,1$. 
Consider a sequence of clinical trials with increasing sample sizes such that $N_k \rightarrow \infty$ for $k=1, \ldots, K$. Then, for any $\bbeta \in \mathbb R^K$, the decision function $\varphi_{\bbeta} : \mathcal{ D} \rightarrow \{0,1\}^K$ 
 in \eqref{eq:procedure1} with weights defined in \eqref{eq:procedure1w} controls asymptotically the FWER at level $\alpha$,
$$\lim_{N\rightarrow \infty} \mathbb E_{\btheta}[h(\varphi_{\beta}(\Data), \btheta)] \leq \alpha,\;\;\;\;\;\;\;\;\;\;\;\;\; \forall \theta\in\Theta,$$
 where $h(\varphi_\beta(\Data), \btheta ) = 1\left \{ \sum_{k=1}^K \varphi_{\beta,k}(\Data) 1 \{ \gamma_k \leq 0\} \geq 1 \right \}$.
 \end{prop}
The proof is provided in the Supplementary Material. The asymptotic result in proposition~\ref{prop:fwer_asymp} suggests that the control of false positives approximately matches the target $\alpha$-level. This justifies the selection of $\bbeta$---to optimize the expected utility---within a subset of designs identified using asymptotic considerations that approximately control the FWER  below the $\alpha$ level. In scenarios with small sample sizes, where asymptotic approximations may be inadequate, we implement a Bootstrap calibration of the optimal decision function (see \secname~\ref{sec:boot_calibration}) to achieve nominal false-positive control.

{\it Monte Carlo optimization approach.} 
To select a parameter $\bbeta \in \mathbb R^K$ that maximizes the utility (\ref{eq:utility_ex1}) in expectation, we use the following approach: 	
\begin{enumerate}
\item Generate independent replicates; $\btheta^{(r)} \sim \pi$ and $\Data^{(r)} \mid \btheta^{(r)} \sim p_{\btheta^{(r)}} $ for $r=1, \ldots, R.$
\item Define the empirical estimate $\widehat{U(\bbeta)} = R^{-1} \sum_{r=1} ^R u( \varphi_{\bbeta}(\Data^{(r)}), \btheta^{(r)})$ of ${\mathbb E}_\pi[u(\varphi_\beta(\Data), \btheta)].$
\item Optimize $\widehat{U(\bbeta)}$ over $\bbeta \in \mathbb R^K$, e.g., using simulated annealing~\citep{belisle:1992} or  a grid search. 
\end{enumerate}
If the utility function $u(\varphi(\Data),\btheta)$ includes large $\lambda$ penalties for false positive results,  then $\varphi^*$, the unconstrained decision-theoretic optimum, may satisfy the regulator's request of an FWER below a target, say $0.05$. In this case, the user can explore different $\alpha$ values between zero and the regulatory threshold to select the decision functions.

\subsection{Results}\label{sec:multipleTestRes}

\textit{PARP inhibitor in patients with breast cancer.} 
We consider a simulation study  that presents similarities with our breast cancer research collaboration~\citep{tung:2020}. The simulation focuses on patients with homologous recombination (HR) defects, a group with a median PFS under 24 months. PARP inhibitors, which target DNA double-strand breaks, have improved PFS, objective response rate (ORR), and quality of life in HR-related cancers, though their efficacy varies by HR alteration. Given strong evidence that ORR is an early surrogate for PFS~\citep{burzykowski:2008}, our simulations consider a clinical trial  that evaluates ORR (auxiliary outcome) at 4 months and PFS (primary outcome) at 18 months, with subgroups defined by gPALB2 and sBRCA1/2a mutations.

\noindent\textit{Simulation model.} 
We conduct a simulation study with $15$ scenarios, varying the degree of    concordance between the treatment effects on  primary and auxiliary outcomes. For each scenario, we generate $5000$ RCTs with binary primary $Y_{k,i}$ and auxiliary $S_{k,i}$ outcomes for $k=1,2 $  subgroups, 
with overall sample size $N=200$. In each of these trials, the probabilities of enrolling patients from the first or second subgroup are $0.6$ and $0.4,$ respectively, with expected sample sizes of $120$ and $80.$

 We simulate outcome data $(Y_{k,i}, S_{k,i})$ for the experimental $(c=1)$ and SOC arms $(c=0)$ conditional on treatment assignments $ C_{k,i} $, for $k=1,\ldots,K,$ and $ i=1, \ldots, N_k,$ using a copula model, varying the marginal outcome probabilities, $\pr(Y_{k,i}=1 \mid C_{k,i}=c)$ and $\pr(S_{k,i}=1 \mid C_{k,i}=c), \; k=1,2, c=0,1,$ and the odds-ratios 
 $$R_{k,c} = \frac{\pr(Y_{k,i} =1, S_{k,i} =1 \mid C_{k,i} = c) ~ \pr(Y_{k,i} =0, S_{k,i} =0\mid C_{k,i} = c) }{\pr(Y_{k,i}=1,S_{k,i}=0\mid C_{k,i} = c)~ \pr(Y_{k,i}=0,S_{k,i} =1\mid C_{k,i} = c) }, 1\leq k\leq K.$$ 

\textit{Simulation scenarios.}
To prepare a study protocol, practical considerations such as expected enrollment and dropout rates, need to be combined with realistic simulations studies. In our setting, investigators expect a positive association between effects on the auxiliary and primary outcomes, and simulations should explore violations of this assumption to evaluate whether relevant operating characteristics are severely compromised or not.
Also, variations in the prior probability model and utility functions can be presented to the stakeholders; we will discuss this later in this section.
Here, we consider $15$ scenarios with different configurations of treatment effects and correlation between primary and auxiliary outcomes. We included five configurations for treatment effects:\\
1) there are no treatment effects for the primary and auxiliary outcomes in both subgroups;
2) there are no treatment effects for the primary outcome in both subgroups, 
but there is a positive treatment effect for the auxiliary outcome in subgroup $k=1$;
3) there is a positive treatment effect for the primary outcome in subgroup $k=1$, 
but there are no treatment effects on the auxiliary outcomes;
4) there are positive treatment effects for both  primary and auxiliary outcomes in subgroup $k=1$;
5) there is a positive treatment effect on the primary outcomes in subgroup $k=1,$ 
but there is a negative treatment effect on the auxiliary outcomes in the same subgroup.

For each of these five configurations, we generate data with various degrees of correlation between the primary and auxiliary outcomes. We set the marginal probability for the control arms to $\pr(Y_{k,i} =1 \mid C_{k,i} =0) = 0.2$, and $\pr(S_{k,i} =1 \mid C_{k,i} =0) = 0.5$ for $k=1,2.$ If there is a positive treatment effect for the primary and/or auxiliary outcome, then $\pr(Y_{k,i} =1 \mid C_{k,i} =1) = 0.4$ and $\pr(S_{k,i} = 1\mid C_{k,i} =1) = 0.75,$  otherwise, the probabilities are identical to the ones for the SOC arm. For configuration 5, with a negative effect on the auxiliary outcomes (subgroup $k=1$), we set for
$\pr(S_{k,i} =1 \mid C_{k,i} =0) = 0.75$ and $\pr(S_i =1 \mid C_{k,i} =1) = 0.5.$ Finally, in our simulation scenarios, the odds-ratios are $R_{k,c}=1, 2,$ or $10$ for subgroups $k=1,2$ and arms $c=0,1.$

 \textit{Hypothesis testing, utility, and decision function.} 
 The treatment effect $\gamma_k$, $k=1,\ldots,K,$ is the mean difference between the primary outcomes under the experimental and control therapies, and $H_{0,k}: \gamma_{k} \leq 0.$ We use the parametric decision function in~\eqref{eq:procedure1} and specify the weights $w_k(\bbeta,\overline{S}) \propto \exp\{ \beta \overline{S}_k \} $ for $k=1,\ldots,K$, with $\beta \geq 0$.

\textit{Joint Bayesian model.} 
We use a prior model to predict the data that the clinical trial will generate, including auxiliary and primary outcomes, and use these predictions to find the parameters of \eqref{eq:procedure1w} that optimize the utility~\eqref{eq:utility_ex1}. %Auxiliary and primary outcomes are modeled jointly.  
Prior information from previous trials or meta-analyses can inform  the design of the trial   \citep[see][for an example]{spring:2020}. Some aspects of the model, such as the relationship between treatment effects on primary and auxiliary outcomes, might present substantial uncertainty, and this should be incorporated into the prior, for example using mixtures~\citep{schmidli:2014}. 

To simplify exposition, we use the same prior model here and in Section~\ref{sec:sequential}. Specifically, we use a Bayesian logistic model with correlated outcomes: 
\begin{align}
\pr(Y_{k,i} =1\mid C_{k,i}, \epsilon_{k,i}, \zeta^Y_{0,k} , \zeta^Y_{1,k}) &= F(\zeta^Y_{0,k} + \zeta^Y_{1,k} C_{k,i} + \epsilon_{k,i}), \nonumber \\
\pr(S_{k,i} =1\mid C_{k,i}, \epsilon_{k,i},\zeta^S_{0,k}, \zeta^S_{1,k} ) &= F(\zeta^S_{0,k} + \zeta^S_{1,k} C_{k,i} + \epsilon_{k,i}), 
\label{eq:multi_logit}
\end{align}
where $F(\cdot)$ is the logistic distribution function, and the independent Gaussian terms $\epsilon_{k,i} \buildrel{\mbox{\tiny iid}} \over \sim N(0, \sigma^2 _{k} )$ modulate the correlation between $Y_{k,i}$ and $S_{k,i}$. An advantage of model~\eqref{eq:multi_logit}, compared to alternative joint models for primary and auxiliary outcomes, is that it requires a single parameter $\sigma^2 _{k}$ to induce correlation. This simplifies prior elicitation. We use independent normal priors for the intercepts $\zeta^Y_{0,k}$ and $\zeta^S_{0,k}.$ Their means and standard deviations, and the parameters $\sigma^2 _{k}$, can be elicited from historical data.  
The prior on $\sigma^2_k$ affects the optimal solutions and their operating characteristics, because it alters the correlations and variances of the primary and auxiliary outcomes, as well as the relationship between the corresponding treatment effects. The value of $\sigma^2_k$ can be elicited by matching the degree of correlation between primary and auxiliary outcomes with published estimates. Alternatively, an inverse-gamma distribution can be used as a prior on $\sigma_k^2$, allowing for additional uncertainty. The prior for the treatment effect $\zeta^S_{1,k}$ on the auxiliary outcome is a two-component mixture:
$$
\xi 
1\{
\zeta^S_{1,k} =0\}
 + (1-\xi) N(m_{S,k}, \sigma^2_{S,k}), 
$$
where $\xi$ is the prior probability that the experimental treatment has no effects. We then specify the prior for the treatment effect $\zeta^Y_{1,k} $ on the primary outcome conditionally on the effect on the auxiliary outcome $\zeta^S_{1,k}$:
$$
 \zeta^Y_{1,k} = c_{Y,k} \zeta^S_{1,k} ~~ \text{ and } ~~
c_{Y,k} \sim \mbox{Beta}( v_{k}, o_{k} ),
$$
where $\mbox{Beta}( a, b )$ indicates the beta distribution with mean $a/(a+b)$ and variance $ab/[(a+b)^2(a+b+1)]$. In other words, the treatment effect on the primary outcome $\zeta^Y_{1,k}$ {\it a priori} correlates with the effect on the auxiliary outcome $\zeta^S_{1,k}$. The values of the hyperparameters are reported in Table~\ref{tab:s1} of the Supplementary Material, together with summaries of trials generated using the prior model. Summaries of the prior, such as the average correlation between primary and auxiliary outcomes within a subgroup of interest, are important for assessing whether the prior is representative---with adequate uncertainty---of historical data from completed trials or electronic health records. We also discuss a variation  of the outlined prior in the Supplementary Material (see \figurename~\ref{fig:s3} and \tablename~\ref{tab:s5}).

 \textit{Testing procedures.}
 We compare the operating characteristics of the proposed procedure (indicated as {\it Auxiliary-Augmented} in \tablename~\ref{tab:sim1}) with (i) {\it Bonferroni} procedure, (ii) {\it Holm} procedure~\citep{holm:1979}, (iii) a testing procedure that only uses the auxiliary outcomes to evaluate the presence of treatment effects with Bonferroni corrections ({\it Auxiliary-Only} in \tablename~\ref{tab:sim1}), and (iv) an implementation of the {\it Frequentist Assisted by Bayes (FAB)} testing procedure proposed in \citet{hoff:2021}. For all the procedures, we target an FWER $\leq 0.05.$ Optimal parameters for the {\it Auxiliary-Augmented} procedure are illustrated in Supplementary \figurename~\ref{fig:s1}. We used $\lambda = 0.5,$ which implies a positive utility if there is at least one true positive result. Simulations with alternative values of $\lambda$ are discussed later in this section. The {\it FAB} procedure defines test statistics and p-values that incorporate indirect or prior information and uses a {\it linking model}; see \citet{hoff:2021} for details. In our implementation of {\it FAB}, we fit logistic regression models for both the auxiliary and primary outcomes. We emphasize that our implementation is not consistent with the recommendations in \citet{hoff:2021} to include independent components of information.  Indeed, primary and auxiliary outcomes correlate. 

\textit{Comparative simulation study.} \tablename~\ref{tab:sim1} shows for each scenario the proportion of times the null hypotheses $H_{0,k}$ were rejected across $5000$ simulations. Supplementary \tablename~\ref{tab:s2} summarizes the FWER.

In scenario 1, without positive treatment effects, all procedures control the FWER at $5\% $ level. {\it FAB} is slightly more conservative with FWER of approximately $4\%$ for $R_{k,c} = 1,2,10.$ In scenario 2, where the auxiliary variable incorrectly indicates a positive treatment effect for group 1, the proportion of rejections of the null hypothesis $H_{0,1}$ is higher for {\it Auxiliary-Augmented} compared to {\it Bonferroni} ($0.039$, $0.036$ and $0.041$ compared to $0.029$, $0.025$, and $0.028$ for $R_{k,c} = 1,2,10$). However, {\it Auxiliary-Augmented} has an FWER close to the nominal level ($0.052$, $0.053$, and $0.059$ for $R_{k,c} = 1,2,$ and $10,$ respectively). Also, in this scenario, {\it FAB} is slightly more conservative than the other two procedures, with FWERs of $0.043$, $0.044$, and $0.044$ for $R_{k,c} = 1,2,$ and $10,$ respectively. As expected, using only the auxiliary outcomes as a surrogate for the primary outcome,  {\it Auxiliary-Only} greatly inflates the FWER ($0.825$, $0.825$, and $0.831$ for $R_{k,c} = 1,2,$ and $10,$ respectively) because in this scenario, there are positive effects of the treatment on the auxiliary outcome but no effects on the primary outcome.

In scenario 3, where the auxiliary outcomes incorrectly suggest no treatment effects in both subgroups, {\it Auxiliary-Augmented}, {\it Bonferroni}, and \textit{FAB} perform similarly. In scenario 4, where the auxiliary variable correctly indicates the presence of a treatment effect for group 1, the proportion of rejections of $H_{0,1}$ (true positive results) for {\it Auxiliary-Augmented} is approximately $0.73$ for $R_{k,c} = 1,2,10$ compared to approximately $0.68$ {\it Holm} and {\it Bonferroni}.  To further interpret this difference, we can estimate the sample size $N$ for a hypothetical trial using the \textit{Bonferroni} procedure that achieves a rejection rate of $H_{0,1}$ of $0.73$ under Scenario 4.  We estimated this total sample size to be approximately $230$, compared to $200$ with \textit{Auxiliary-Augmented}. {\it FAB} has a lower proportion of simulation replicates that reject $H_{0,1}$ than the remaining procedures (e.g, $0.66$ compared to $0.73$ with {\it Auxiliary-Augmented} for $R_{k,c} = 10$). In scenario 5, there is a positive treatment effect on the primary outcomes in subgroups $k=1,$ but a negative effect for the auxiliary outcome in this subgroup. In this scenario, as expected, leveraging information on auxiliary treatment effects (our {\it Auxiliary-Augmented} procedure) leads to a reduction in the proportion of rejections of $H_{0,1}$ compared to the {\it Holm}, {\it Bonferroni}, and {\it FAB}. For example, in the setting with $R_{k,c} =10$ {\it Auxiliary-Augmented} rejects $H_{0,1}$ with a proportion of approximately $0.57$ compared to $0.68$ for the {\it Holm}. 

\textit{Sensitivity Analyses.}
Comparing how the operating characteristics of alternative \textit{Auxiliary-Augmented} procedures vary when optimized under different values of $\lambda$ is an important sensitivity analysis. Frequentist constraints are built into the optimization, so are expected to be robust to the choice of $\lambda$, but other operating characteristics, such as power, may vary substantially. For the 15 simulation scenarios described in this section, we considered ten values of the parameter $\lambda\in \{0.1,0.2,0.3,\cdots, 1\}$, and computed the same operating characteristics as in  \tablename~\ref{tab:sim1}. These $\lambda$ values leads to  optimal values for $\beta$ between $14.07$ and $3.69.$  The resulting operating  characteristics are reported in \tablename~\ref{tab:s3} of the Supplementary Material.

The FWER for simulation scenarios 1 and 2,  with $R_{k,c} =1,2,10$, is between $0.051$ and $0.065$. In scenario 3, where the auxiliary outcomes incorrectly suggest no treatment effects in both subgroups, for increasing values of $\lambda,$ $H_{0,1}$ is rejected in a proportion of the simulations that varies between $0.63$ and $0.67$, for all $R_{k,c} =1,2,10.$ In scenario 4, where the auxiliary variable correctly indicates the presence of a treatment effect for group 1, $H_{0,1}$ is rejected in a proportion of the simulations between $0.73$ and $0.76.$

We observed larger variations of the operating characteristic in scenario 5, with a positive treatment effect on the primary outcomes in subgroup 1 but a negative effect on the auxiliary outcome in this subgroup. In this scenario, $H_{0,1}$ is rejected in a proportion of the simulations between $0.34$ to $0.59$ (\tablename~\ref{tab:s3}). As expected, with larger values of $\beta$ the frequency of false negative results increases.

Another important sensitivity analysis involves considering variations of the prior model. In \figurename~\ref{fig:s3} and \tablename~\ref{tab:s5} of the supplementary, we considered a variation of the prior model described in Section~\ref{sec:multipleTestRes}. Compared to Section~\ref{sec:multipleTestRes}, the prior changes the distribution of $c_{Y,k}$, for $k = 1, 2$, to a mixture of two components, a beta distribution with parameters $v_k = 6$ and $o_k = 1$ with probability $0.9$, and a point mass at zero with probability $0.1$. In this case, the optimal design is parametrized by $\beta_{Y,S}$ = $1.98$.

\textit{Simulation study with $K=6$ subgroups.} 
\tablename~\ref{tab:sim_k6} in the Supplementary Materials reports results of additional simulations with $K=6$ subgroups. We consider the same $15$ simulation scenarios described in this section. The data distribution for the first ($k=1$) subgroup is identical to the previous simulations, and the data for subgroups $k = 2, \ldots, 6$ are generated using the same models we used for the second subgroup.
We generate $5000$ RCTs for each scenario with sample size $N=600$. In each trial, the probability of enrolling a patient to the first subgroup is $0.25$, and $0.15$ for subgroups $2,\ldots,6$, with an expected sample size of $150$ for the first subgroup, and $90$ for subgroups $2,\ldots,6.$ Optimal parameters for the {\it Auxiliary-Augmented} procedure are illustrated in Supplementary \figurename~\ref{fig:s2}. We used  $\lambda = 1,$ which implies a positive utility if the number of true positive findings is larger than the number of false positive results.

In scenarios 1 and 2, all methods present FWER close to the targeted $0.05$, with {\it Auxiliary-Augmented} being slightly anti-conservative (FWER $\sim 6\%)$ and {\it FAB} slightly conservative (FWER $\sim 3\%).$ We note that {\it Auxiliary-Augmented} with the bootstrap calibration procedure described in \secname~\ref{sec:boot_calibration} controls the FWER at the nominal level. In scenario 3, where the auxiliary outcomes incorrectly suggest no treatment effects in all the subgroups, the performance of {\it Auxiliary-Augmented} is slightly worse than the other methods. For example, with $R_{k,c} = 2,$ {\it Auxiliary-Augmented} rejects $H_{0,1}$ in $56.5\%$ of the simulation replicates compared to $63.8\%$ with the {\it Holm} procedure. In scenario 4, where the auxiliary variable correctly suggests the presence of a treatment effect for group 1, we observe a higher proportion of $H_{0,1}$ rejections (true positives) with  {\it Auxiliary-Augmented} compared to {\it Holm}. For example, when $R_{k,c}=1$, {\it Auxiliary-Augmented} rejects $H_{0,1}$ approximately $80\%$ of the time, compared to $63\%$ for {\it Holm}. In scenario 5, where treatment effects in subgroups $k=1$ are discordant, the performance of  {\it Auxiliary-Augmented} deteriorates, as expected. The proportion of replicates with the rejection of $H_{0,1}$ is approximately $0.28$ for $R_{k,c} =1,2,10,$ compared to $0.62$ of {\it Holm}. 
\begin{table}[h!]
 \centering
 \resizebox{0.94\textwidth}{!}{\begin{tabular}{rcccccc} 
\toprule 
\multicolumn{7}{c}{Scenario 1}\\ 
\midrule 
& \multicolumn{2}{c}{$ R_{k,c}= 1$} & \multicolumn{2}{c}{$R_{k,c} = 2$} & \multicolumn{2}{c}{$R_{k,c} = 10$} \\ 
Method  & subgroup 1 & subgroup 2 & subgroup 1 & subgroup 2 & subgroup 1 & subgroup 2 \\ 
Auxiliary-Augmented &0.023 & 0.029 &0.030 & 0.030 &0.030 & 0.028 \\ 
Auxiliary-Augmented-B &0.023 & 0.028 &0.028 & 0.026 &0.025 & 0.024 \\ 
Bonferroni &0.024 & 0.028 &0.028 & 0.029 &0.026 & 0.023 \\ 
Holm &0.025 & 0.029 &0.029 & 0.029 &0.027 & 0.024 \\ 
FAB &0.021 & 0.022 &0.025 & 0.021 &0.024 & 0.019 \\ 
Auxiliary-Only &0.027 & 0.026 &0.024 & 0.028 &0.029 & 0.025 \\ 
\bottomrule 
\multicolumn{7}{c}{Scenario 2}\\ 
\midrule 
& \multicolumn{2}{c}{$ R_{k,c} = 1$} & \multicolumn{2}{c}{$ R_{k,c} = 2$} & \multicolumn{2}{c}{$ R_{k,c} = 10$} \\Method  & subgroup 1 & subgroup 2 & subgroup 1 & subgroup 2 & subgroup 1 & subgroup 2 \\ 
Auxiliary-Augmented &0.039 & 0.014 &0.036 & 0.017 &0.041 & 0.019 \\ 
Auxiliary-Augmented-B &0.038 & 0.015 &0.034 & 0.016 &0.035 & 0.016 \\ 
Bonferroni &0.029 & 0.025 &0.025 & 0.025 &0.028 & 0.026 \\ 
Holm &0.029 & 0.026 &0.026 & 0.026 &0.029 & 0.027 \\ 
FAB &0.025 & 0.018 &0.024 & 0.020 &0.023 & 0.022 \\ 
Auxiliary-Only &0.820 & 0.026 &0.820 & 0.029 &0.825 & 0.026 \\ 
\bottomrule 
\multicolumn{7}{c}{Scenario 3}\\ 
\midrule 
& \multicolumn{2}{c}{$ R_{k,c} = 1$} & \multicolumn{2}{c}{$R_{k,c} = 2$} & \multicolumn{2}{c}{$R_{k,c} = 10$} \\ 
Method  & subgroup 1 & subgroup 2 & subgroup 1 & subgroup 2 & subgroup 1 & subgroup 2 \\ 
Auxiliary-Augmented &0.673 & 0.026 &0.672 & 0.031 &0.666 & 0.036 \\ 
 Auxiliary-Augmented-B &0.663 & 0.026 &0.652 & 0.027 &0.630 & 0.031 \\ 
 Bonferroni &0.678 & 0.027 &0.675 & 0.029 &0.678 & 0.031 \\ 
Holm &0.682 & 0.041 &0.678 & 0.043 &0.683 & 0.050 \\ 
FAB &0.664 & 0.019 &0.663 & 0.022 &0.664 & 0.023 \\ 
Auxiliary-Only &0.024 & 0.030 &0.027 & 0.027 &0.028 & 0.028 \\ 
\bottomrule 
\multicolumn{7}{c}{Scenario 4}\\ 
\midrule 
& \multicolumn{2}{c}{$R_{k,c} = 1$} & \multicolumn{2}{c}{$R_{k,c} = 2$} & \multicolumn{2}{c}{$R_{k,c} = 10$} \\ 
Method  & subgroup 1 & subgroup 2 & subgroup 1 & subgroup 2 & subgroup 1 & subgroup 2 \\ 
Auxiliary-Augmented &0.731 & 0.014 &0.735 & 0.019 &0.731 & 0.016 \\ 
 Auxiliary-Augmented-B &0.726 & 0.012 &0.715 & 0.017 &0.701 & 0.015 \\ 
 Bonferroni &0.681 & 0.022 &0.677 & 0.031 &0.676 & 0.024 \\ 
Holm &0.683 & 0.039 &0.680 & 0.048 &0.679 & 0.045 \\ 
FAB &0.669 & 0.017 &0.662 & 0.024 &0.665 & 0.018 \\ 
Auxiliary-Only &0.825 & 0.029 &0.814 & 0.024 &0.812 & 0.025 \\ 
\bottomrule\multicolumn{7}{c}{Scenario 5}\\ 
\midrule 
& \multicolumn{2}{c}{$R_{k,c} = 1$} & \multicolumn{2}{c}{$R_{k,c} = 2$} & \multicolumn{2}{c}{$R_{k,c} = 10$} \\ 
Method  & subgroup 1 & subgroup 2 & subgroup 1 & subgroup 2 & subgroup 1 & subgroup 2 \\ 
Auxiliary-Augmented &0.579 & 0.037 &0.565 & 0.041 &0.569 & 0.043 \\ 
 Auxiliary-Augmented-B &0.574 & 0.035 &0.552 & 0.038 &0.538 & 0.037 \\ 
 Bonferroni &0.683 & 0.026 &0.673 & 0.029 &0.676 & 0.028 \\ 
Holm &0.686 & 0.040 &0.676 & 0.045 &0.678 & 0.044 \\ 
FAB &0.671 & 0.020 &0.660 & 0.019 &0.661 & 0.022 \\ 
Auxiliary-Only &0.000 & 0.029 &0.000 & 0.028 &0.000 & 0.026 \\ 
\bottomrule 
\end{tabular}}
 \caption{Proportion of times the null hypotheses of absent treatment effects on the primary outcome is rejected across $5000$ simulations. The Auxiliary-Augmented-B method refers to the bootstrap calibrated procedure described in \secname~\ref{sec:boot_calibration}.
 }
 \label{tab:sim1}
\end{table}

Study teams and biostatisticians can use the utility function~\eqref{eq:utility_ex1} with parameter $\lambda$ selected to ensure a positive utility when the number of true positive results exceeds the number of false positives. In our simulations, $\lambda = 0.5$ when $K=2$, and $\lambda =1$ when $K=6$. We tested the prior model described in Section~\ref{sec:multipleTestRes}, with hyperparameters listed in \tablename~\ref{tab:s1},  across a broad range of scenarios. As expected, substantial improvements in operating characteristics had been apparent in scenarios where treatment effects concentrate in a small fraction of subgroups. More generally, with these choices of utility criteria and modeling, when the likelihood of discordant treatment effects between primary and auxiliary outcomes is minimal, we expect moderate-to-large improvements of operating characteristics. However, if the likelihood of discordant effects is moderate to high, we recommend a comprehensive assessment of hypothetical scenarios, including potential discrepancies between the distribution of actual outcomes and the Bayesian model, based on simulations.

%\clearpage

\subsection{Bootstrap calibration of the decision function}\label{sec:boot_calibration}
Proposition~\ref{prop:fwer_asymp} states that for any $\boldsymbol{\beta} \in \mathbb R^K$ of the decision function, the {\it Auxiliary-Augmented} testing procedure controls the FWER asymptotically at level $\alpha$. For RCTs with small sample sizes and strong correlations between $\overline Y_k$ and $\overline S_k$, the finite sample FWER might not match  the nominal $\alpha$ level, see Supplementary \figurename~\ref{fig:boot} for an example. Intuitively, if the statistics $\overline S_k$ and $\overline Y_k$ present a positive correlation, then positive $\overline S_k$ values will be associated with small p-value ($\text{\tt pv}_k$) for the group $k$ and large weights $\omega_k(S_k, \beta)$. This mechanism can inflate the FWER.

With small sample sizes, we can estimate the FWER of the optimal solution via simulation by generating correlated p-values ($\text{\tt pv}_k$) and weights $\omega_k(S_k, \beta)$ in the absence of treatment effects on the primary outcome. This estimation-based (bootstrap) strategy can be used to correct for potential inflation of the FWER above the targeted $\alpha$ value.
 
Specifically, the optimal testing procedure is calibrated by estimating the threshold $\alpha'$, that matches the FWER of the adjusted decision functions $\varphi^{\text{\tt adj}}_{\bbeta,k}(\Data) =
1\{ \text{\tt pv}_k \leq \alpha^\prime \omega_k( \bbeta , \overline{S}) \}$ and the nominal $\alpha$-level. For simplicity, we use $\bbeta$ instead of $\bbeta_{ Y, S}$ for the parameter of the optimal parametric decision in expression \eqref{eq:par_maximization}:\\

\noindent Input:
\begin{enumerate}
\item[(i)] Auxiliary treatment effect estimates $\{ \widetilde S_k \}_{k=1}^K$. 
\item[(ii)] Estimated covariance matrices $\widetilde \Sigma_k$
of the random vectors $(\bar S_k, \bar Y_k)$ for $k=1, \ldots, L$.
\item[(iii)] Parameter $\bbeta$. 
\end{enumerate}
\noindent Procedure:
\begin{enumerate}
\item Generate $b=1, \ldots, B$ data summaries 
$(\overline Y^{(b)}_k,\overline S^{(b)}_k) \sim N((0,\widetilde S_k), \widetilde{\Sigma}_k)$.
\item For each $k=1,\ldots,K$ and $b=1,\ldots,B$,
compute the weights $\omega^{(b)}_k = \omega_k(\boldsymbol \beta,\overline S^{(b)})$ 
and the p-values $\text{\tt pv}_k^{(b)}$, 
which are functions of 
$\overline Y^{(b)}_k$ and $\overline S^{(b)}_k$. \item Set $\alpha^\prime = \inf \{ t \in [0,1] : \widehat {FWER} (t) \geq \alpha\}$, where 
$$
\widehat {FWER} (t) = \sum_{b=1}^B G^{(b)}(t) / B, ~ \text{ with } G^{(b)}(t) = 1-\prod_{k=1}^K1\left\{\text{\tt pv}_k^{(b)} \geq \omega^{(b)}_k t\right\}.
$$
\item Conduct the hypothesis tests using the actual data and 
$\varphi^{\text{\tt adj}}_{ \bbeta ,k}(\Data) = 1\{ \text{\tt pv}_k \leq \alpha^\prime \omega_k( \bbeta , \overline{S}) \}$, $k=1, \ldots K.$
\end{enumerate}
The input parameters $\{ \widetilde S_k ,\widetilde \Sigma_k \}_{k=1}^K$ can be set using any consistent estimator of the treatment effects and covariance matrices, such as maximum likelihood or Bayesian estimates. 
\figurename~\ref{fig:boot} shows that the outlined bootstrap calibration algorithm controls the FWER at the nominal $\alpha$-level. 
Moreover, for small odds ratios ($R_{k,c} = 1$) we observe minor reductions in power of approximately $0.5\%$; with $R_{k,c} = 10,$ the power reductions are approximately $ 3\%$ (see scenarios 3-4 in \tablename~\ref{tab:sim1}).

The described calibration was applied only during the analysis stage, and it is recommended when the correlation between primary and auxiliary outcomes is moderate-to-large. In our simulations, its inclusion in the optimizations of $\beta$ had a negligible impact. More generally, the need for calibration or adjustment depends on sample size, the class of parametric decision functions, and other trial features, and the procedure can be embedded within the optimization used to select the optimal parametrization.

\section{Auxiliary outcomes for interim decisions in RCTs}\label{sec:sequential}
In this section, we discuss a multi-stage design motivated by our work in newly diagnosed glioblastoma (nGBM) \citep{vanderbeek2018clinical, vanderbeek2019randomize}, an aggressive form of brain cancer. The SOC in nGBM is temozolomide in combination with radiation therapy (TMZ+RT), which was approved more than 15 years ago based on the results of the \textit{EORTC-NCIC CE.3} trial \citep{stupp:2005}. Clinical outcomes remain poor, with a median overall survival of less than $16$ months. All confirmatory phase III RCTs conducted since 2005 have failed to demonstrate OS improvements compared with TMZ+RT. Several primary and co-primary outcomes have been used in recent nGBM trials (e.g., ORR, OS, PFS, and OS at 24 months). Investigators in nGBM have highlighted the need for novel drug development strategies and effective trial designs that can rapidly discontinue evaluation of toxic or inferior experimental treatments without compromising power \citep[e.g.,][]{ventz:2021}. Several authors have proposed prediction procedures that trigger futility-stopping decisions, and some of the resulting approaches are currently used in clinical research. 

We hypothesize that combining information on primary and auxiliary outcomes can improve interim decision-making and expose fewer patients to experimental treatments without positive effects. Additionally, with improved interim predictions, the study's power can increase by reducing the likelihood of early futility stopping when the experimental treatment has a positive effect on the primary outcome. Recent contributions have shown that early auxiliary outcomes, which are predictive of the primary outcomes, can be useful to estimate the conditional power~\citep {li:2022,li:2023}. Here, we leverage primary and auxiliary outcomes to make early futility decisions during the trial in accordance with explicit utility criteria. We consider a group-sequential trial design \citep{demets:1994} that uses OS after 18 months of treatment (OS-18, $Y_i\in \{0,1\}$) as primary outcome, and ORR ($S_i\in \{0,1\}$) according to the iRANO criteria~\citep{wen:2010} as auxiliary outcome. 

{\it Study design.} 
The design has $T$ stages and evaluates efficacy in the overall population ($K = 1$), with a maximum of $2\times(T-1)$ interim decisions based on available evidence for futility or efficacy plus a potential efficacy test at the end of the trial (i.e., $M=2(T-1)+1$). Each IA can (i) reject the null hypothesis $H_0$ of no treatment effect on the primary outcome if there is sufficient early evidence and close the study, (ii) terminate the trial early for futility if additional data are unlikely to demonstrate positive effects of the experimental treatment on the primary outcomes, or (iii) continue the study. At the end of the $t$-th stage ($t=1,\ldots,T-1$), the design includes an efficacy analysis, followed by a futility analysis unless $H_0$ was rejected. The $t$-th analysis occurs after the primary outcomes of the first $n_t$ patients become available, exactly 18 months after the enrollment of the $n_t$-th patient, where  $n_1,\ldots, n_T, (n_t<n_{t+1})$ are fixed design parameters. Potentially more than $n_t$ auxiliary outcomes are available at IA $t$ and are used for interim decisions. Lastly, we use $m_t, t\leq T, m_t \geq n_{t}$ to indicate the number of enrollments before the $t$-th  analysis. 
 
{\it Frequentist constraint and utility function.} The type I error rate is controlled at level $\alpha$. We indicate with $\varphi_{t, E}: \Data^t \rightarrow \{0,1\}$, $1\leq t \leq T$, the efficacy decisions and with $\varphi_{t, F}: \Data^t \rightarrow \{0,1\}$, $1\leq t \leq T-1$, the futility decisions. Let $T^{\star} \in \{1,\ldots,T\}$  be the time when the trial stops. We use the utility function
\begin{eqnarray}
u(\varphi(\Data), \btheta) &=&  \lambda^{'}_{T^{\star}} 
 1\Big\{\gamma(\btheta)>0, \varphi_{T^\star,E}(\Data^{T^\star})=1\Big \}
 -\lambda m_{T^{\star}},
 \label{eq:utility_seq}
\end{eqnarray}
with $\varphi(\Data) =\{ (\varphi_{t,E} (\Data^t), \varphi_{t,F} (\Data^{t}) ), t=1,\ldots,T-1, \varphi_{T,E} (\Data^T) \},$
  $\lambda >0,$ and $\lambda^{'}_{t} >0$, for $t=1,\ldots,T$.
 This utility function assigns a reward $\lambda^{'}_{T^{\star}}
 $ for correctly rejecting the null hypothesis $H_{0}$ at time $T^{\star}$, and a cost $\lambda$ for each  enrolled patient. 

\textit{Parametric efficacy decision function.} There are several popular sequential decision rules to stop a trial for efficacy, including the Pocock \citep{pocock:1977},  O’Brien–Fleming \citep{obrien:1979} stopping rules, and the alpha-spending functions~\citep{lan:1983}. We specify the efficacy stopping rule $\varphi_{t,E} (\Data^t)$, parameterized by $\beta_{E} \in \mathbb R,$ using an alpha-spending function.  For binary outcomes, the alpha-spending function $\alpha^\star (r)$ for $r\in [0,1]$ is a function of the information fraction $r_t=n_t/n_{T}.$  It determines the proportion of the overall type I error rate $\alpha$ that is spent by the $t$-th IA. The type-I error allocated to the  $t$-th IA is $\bar{\alpha}(r_t) = \alpha^\star(r_t)$ - $\alpha^\star(r_{t-1}), 1\leq t \leq T$, where $r_0=0;$ we refer to \citet{lan:1983} and \citet{demets:1994} for details. We use the family in~\citet{hwang:1990} with $\beta_{E} \in \mathbb{R}$:
\begin{equation}
\alpha^{\star}_{\beta_{E}} (r) = 
\begin{cases}
\alpha \frac{1-\exp\{-\beta_{E}\times r\}}{1-\exp\{-\beta_{E} \}} & \mbox{if } \beta_{E} \neq 0, \\
\alpha \times r & \mbox{if } \beta_{E} = 0, 
\end{cases}
\label{eq:alpha_spending}
\end{equation}
  
As in \secname~\ref{sec:weighted_bonferroni}, we use efficacy decisions based on Z-statistics. Let $Z_t$ be the $Z$-statistics based on the primary outcomes of the first $n_t$ patients. We indicate with $z_{\beta_{E}, t}$ the recursively defined $(1-\bar{\alpha}_{\beta_{E}} (r_t))$-quantiles of the conditional distributions $p ( Q_t \mid Q_{t'} \leq z_{\beta_{E}, t'}, 1 \leq t' < t)$, where $(Q_1, \ldots, Q_T)$ is a mean zero Gaussian vector with covariance $Cov(Q_t, Q_{t'} )=\sqrt{n_t/n_{t'}},$ for $t\le t'$. The resulting $z_{\beta_{E}, t}$ values are the Z-statistics thresholds for the efficacy decisions (i.e., $\varphi_{t,E} (\Data^t)$). 
See \citet{jennison:1999} for a detailed explanation of how to compute the vector $(z_{\beta_{E}, 1}, \ldots, z_{\beta_{E},T}).$ 

The proposed alpha-spending function is based solely on the primary outcomes. This avoids a scenario in which an early auxiliary outcome suggests a treatment effect and most of the available type I error ($\alpha$) is consumed at the first interim analysis. If the interim efficacy analysis is not significant---potentially due to limited power---the small remaining $\alpha$ makes it unlikely that later analyses will be successful.

 \textit{Parametric futility decision function.} 
 Several approaches for early futility stopping use rules to stop the trial at IA $t$ if, given the available interim data, the likelihood of obtaining a statistically significant study result at any stage $t< t' \leq T$ falls below a pre-defined threshold~\citep{lan:1982, betensky:1997b, lachin:2005, berry:2010}. The likelihood of a positive result $\pr( \cup_{t< t' \leq T} \{ Z_{t'} > z_{\beta_{E}, t'}\} \mid \Data^t)$ can be computed using frequentist~\citep{lan:1982, betensky:1997b, lachin:2005} or Bayesian methodologies~\citep{spiegelhalter:1986,berry:2010}. Most of these methods approximate the probability $\pr( \cup_{t< t' \leq T} \{ Z_{t'} > z_{\beta_{E}, t'}\} \mid \Data^t)$ through computations that ignore the possibility of stopping for futility at time $t'>t$, that is all future futility IAs are ignored. We specify the futility stopping  criteria (i.e., $\varphi_{t,F} (\Data^{t})$) using Bayesian modeling. Our approach leverages auxiliary and primary outcomes into predictions and futility decisions. We stop for futility $\varphi_{t,F} (\Data^{t}) = 1$, when $\pr( \cup_{t< t' \le T} \{ Z_{t'} > z_{\beta_{E}, t'}\} \mid \Data^t) \leq \beta_{F}$. Note that $\varphi$ is parametrized by $\beta = (\beta_{E},\beta_{F}).$ The conditional probability $\pr( \cup_{t< t' \leq T} \{ Z_{t'} > z_{\beta_{E}, t'}\} \mid \Data^t),$ that we use for futility decisions, is based on a joint model of primary and auxiliary outcomes. We compute these probabilities using the Hamiltonian Monte Carlo implemented in  \texttt{rstan}~\citep{rstan}.

\textit{Prior model.} 
We use the same prior model as in \secname~\ref{sec:weighted_bonferroni} ($K=1$). 
Optimal values for $\beta=(\beta_F,\beta_E)$ are obtained  by maximizing the utility~\eqref{eq:utility_seq}, { with  $\lambda^{'}_{1}  =1,$  $\lambda^{'}_{2} = 0.5$, and a cost of $\lambda = 0.00005,$ for each patient.}  See also Supplementary \figurename~\ref{fig:utlity_sec6} for  details. 
 
\textit{Simulation Scenarios.}
We consider a two-stage RCT ($T=2$) with a maximum sample size of $N=200$ and $\alpha = 0.05.$ 
We generated the data using the same model of \secname~\ref{sec:multipleTestRes} with $K=1$, as well as the same simulation scenarios (population $k=1$). 

\textit{Alternative trial designs.}
We compare our approach (\textit{Auxiliary-Augmented} in \tablename~\ref{tab:sec6_tab}) with two alternative trial designs. The first is the  {\it Primary-Only} design, which uses only the primary outcome data for efficacy and futility analyses. This design uses the same group-sequential efficacy rule as the \textit{Auxiliary-Augmented} design based on the $\alpha$-spending \eqref{eq:alpha_spending} with $\beta_{E} = 2$. The futility decisions are based on a Bayesian model that ignores auxiliary information, a variation of model~\eqref{eq:multi_logit}. Specifically, we use the model  $\pr(Y_{k,i} =1\mid C_{k,i}, \zeta^Y_{0,k} , \zeta^Y_{1,k}) = F(\zeta^Y_{0,k} + \zeta^Y_{1,k} C_{k,i}),$ with a Gaussian prior on $\zeta^Y_{0,k}$, and  the mixture $\xi 1\{ \zeta^Y_{1,k} =0\} + (1-\xi) N(m_{Y,k}, \sigma^2_{Y,k})$ for $\zeta^Y_{1,k},$ and the same futility threshold $\beta_{F}$ of the \textit{Auxiliary-Augmented} design. The other design in our comparisons is the \textit{Auxiliary-Only} design. It is nearly identical to the \textit{Primary-Only} design, but it replaces the primary outcomes with the auxiliary outcomes for efficacy and futility decisions.

\textit{Results.}
For each simulation scenario, \tablename~\ref{tab:sec6_tab} reports the proportion of simulations in which the null hypothesis $H_0$ of no treatment effects is rejected. It also reports the frequency of simulations in which $H_0$ is rejected at IA and FA, and the average sample size. In these simulations, we did not model the patient enrollment times.

In scenario 1, without positive treatment effects, the three designs control the type-I error at $5\%.$ The \textit{Auxiliary-Augmented} design has the lowest expected sample size with approximatively $108$ patients compared to $118$ and $127$ of the \textit{Primary-Only} and \textit{Auxiliary-Only} designs.
In other words, \textit{Auxiliary-Augmented} provides a $9\%$ reduction in sample size compared to \textit{Primary-Only} and a $16\%$ compared to \textit{Auxiliary-Only}.

In scenario 2, the auxiliary outcomes suggest a positive effect,  but the experimental therapy does not improve survival. Not surprisingly, the \textit{Auxiliary-Only} design has a highly inflated type-I error rate. The \textit{Auxiliary-Augmented} and \textit{Primary-Only} designs control the type I error at the nominal level. In these simulations, the \textit{Auxiliary-Augmented}  design, as expected, stops the trial for futility less frequently than the  \textit{Primary-Only} design; the resulting average sample sizes are $138$  and  $118$ patients. 
 
In scenario 3, there is a positive effect on survival, but the effect on the auxiliary outcomes is null. The \textit{Auxilairy-Only} design does not detect the treatment effect; the power is  $5\%$. The \textit{Primary-Only} design has a power of $87\%$ and an expected sample size of approximately $127$ patients, compared to the \textit{Auxiliary-Augmented} design with a power of approximately $79\%$ and an average sample size of $116$ patients. 

In scenario 4, there is a positive treatment effect on both primary and auxiliary outcomes. The \textit{Auxiliary-Only} design has a power of approximately $96\%$ and an average sample size of $120$ patients. In comparison, the \textit{Primary-Only} design has a power of $87\%$ and an average sample size of $127$ patients, and the \textit{Auxiliary-Augmented} design has a power of  $89\%$ with an average sample size of $131$ patients. 

In scenario 5, the treatment has a positive effect on the primary outcome and a negative effect on the auxiliary. The \textit{Auxiliary-Only} design does not detect the treatment effect. The \textit{Primary-Only} design has a power of approximately $87\%$ and an average sample size of $127$ patients, while the \textit{Auxiliary-Augmented}  design, as expected, presents a reduced power of approximately $52\%$,  due to early futility decisions,   and an average sample size of  $101$ patients. 
\begin{table}[h!]
\center
\resizebox{\textwidth}{!}{ 
\begin{tabular}{r|cc|cc|cc} 
\toprule 
\multicolumn{7}{c}{Scenario 1}\\ 
\midrule 
& \multicolumn{2}{c}{$R_{1,c}= 1$} & \multicolumn{2}{c}{$R_{1,c} = 2$} & \multicolumn{2}{c}{$R_{1,c} = 10$} \\ 
Design &  False positive results &   $\mathbb E[N]$ &
               False positive results &   $\mathbb E[N]$ &
               False positive results &   $\mathbb E[N]$  \\ 
Auxiliary-Augmented &0.048  (0.042; 0.006) & 107.6 &0.043  (0.036; 0.008) & 108.3 &0.051  (0.042; 0.009) & 109.2 \\ 
Primary-Only &0.054  (0.043; 0.011) & 118.2 &0.049  (0.036; 0.013) & 117.8 &0.055  (0.043; 0.012) & 118.4 \\ 
Auxiliary-Only &0.060  (0.046; 0.014) & 126.7 &0.047  (0.038; 0.010) & 125.9 &0.057  (0.041; 0.015) & 128.3 \\ 
\bottomrule 
\multicolumn{7}{c}{Scenario 2}\\ 
\midrule 
& \multicolumn{2}{c}{$R_{1,c}= 1$} & \multicolumn{2}{c}{$R_{1,c} = 2$} & \multicolumn{2}{c}{$R_{1,c} = 10$} \\ 
Design &  False positive results &   $\mathbb E[N]$ &
               False positive results &   $\mathbb E[N]$ &
               False positive results &   $\mathbb E[N]$  \\ 
Auxiliary-Augmented &0.055  (0.039; 0.016) & 137.5 &0.055  (0.039; 0.015) & 136.9 &0.053  (0.039; 0.014) & 138.1 \\ 
Primary-Only &0.053  (0.039; 0.014) & 118.2 &0.052  (0.039; 0.012) & 117.6 &0.051  (0.039; 0.013) & 118.8 \\ 
Auxiliary-Only &0.958  (0.791; 0.167) & 119.6 &0.956  (0.788; 0.168) & 119.6 &0.961  (0.792; 0.169) & 119.5 \\ 
\bottomrule 
\multicolumn{7}{c}{Scenario 3}\\ 
\midrule 
& \multicolumn{2}{c}{$R_{1,c}= 1$} & \multicolumn{2}{c}{$R_{1,c} = 2$} & \multicolumn{2}{c}{$R_{1,c} = 10$} \\ 
Design &  Power &   $\mathbb E[N]$ &
               Power &   $\mathbb E[N]$ &
               Power &   $\mathbb E[N]$  \\ 
Auxiliary-Augmented &0.788  (0.661; 0.126) & 116.5 &0.797  (0.661; 0.137) & 117.0 &0.791  (0.677; 0.114) & 113.6 \\ 
Primary-Only &0.866  (0.664; 0.202) & 127.6 &0.873  (0.664; 0.209) & 127.9 &0.873  (0.680; 0.194) & 125.9 \\ 
Auxiliary-Only &0.052  (0.040; 0.013) & 127.5 &0.051  (0.040; 0.012) & 178.1 &0.053  (0.039; 0.014) & 129.1 \\ 
\bottomrule 
\multicolumn{7}{c}{Scenario 4}\\ 
\midrule 
& \multicolumn{2}{c}{$R_{1,c}= 1$} & \multicolumn{2}{c}{$R_{1,c} = 2$} & \multicolumn{2}{c}{$R_{1,c} = 10$} \\ 
Design &  Power &   $\mathbb E[N]$ &
               Power &   $\mathbb E[N]$ &
               Power &   $\mathbb E[N]$  \\ 
Auxiliary-Augmented &0.892  (0.673; 0.219) & 130.7 &0.887  (0.666; 0.221) & 130.9 &0.892  (0.666; 0.226) & 130.5 \\ 
Primary-Only &0.876  (0.673; 0.204) & 126.9 &0.870  (0.666; 0.204) & 127.4 &0.876  (0.666; 0.210) & 127.4 \\ 
Auxiliary-Only &0.959  (0.805; 0.154) & 118.3 &0.961  (0.790; 0.170) & 119.6 &0.962  (0.789; 0.173) & 120.2 \\ 
\bottomrule 
\multicolumn{7}{c}{Scenario 5}\\ 
\midrule 
& \multicolumn{2}{c}{$R_{1,c}= 1$} & \multicolumn{2}{c}{$R_{1,c} = 2$} & \multicolumn{2}{c}{$R_{1,c} = 10$} \\ 
Design &  Power &   $\mathbb E[N]$ &
               Power &   $\mathbb E[N]$ &
               Power &   $\mathbb E[N]$  \\ 
Auxiliary-Augmented &0.519  (0.511; 0.008) & 101.0 &0.540  (0.533; 0.008) & 101.0 &0.519  (0.514; 0.006) & 100.6 \\ 
Primary-Only &0.877  (0.668; 0.209) & 127.1 &0.876  (0.678; 0.198) & 126.5 &0.868  (0.663; 0.205) & 127.4 \\ 
Auxiliary-Only &0.000  (0.000; 0.000) & 100.1 &0.000  (0.000; 0.000) & 100.1 &0.000  (0.000; 0.000) & 100.1 \\ 
\bottomrule 
\end{tabular} 

}
\caption{Proportion of times the null hypothesis $H_0$ of no treatment effect is rejected (Power) and the expected sample size ($\mathbb E [N]$) for 5000 simulated trials. In parentheses,  we report the frequency of simulations in which $H_0$ is rejected at the interim and final analyses, respectively.
}
\label{tab:sec6_tab}
\end{table}
\subsection{Retrospective analysis with GBM trial data} \label{sec:data_app}
We use patient-level data from the CENTRIC trial~\citep[][]{centric} to evaluate the {\it Auxiliary-Augmented} trial design. CENTRIC is a Phase III clinical trial that enrolled patients with nGBM and tumors with methylated  O6-methylguanine (MGMT). Data on  $273$ patients treated with the SOC (TMZ+RT) are available through \textit{Project Data Sphere}~(\url {https://data.projectdatasphere.org/}). Our analysis uses OS at 24 months (OS-24) as the binary primary outcome and PFS at 12 months (PFS-12) as the auxiliary outcome. We excluded  $34$ patients with unknown 24-month survival status from the analysis.

We consider an RCT with a control arm (TMZ+RT) and an experimental arm,  a total sample size of 
$200$ patients and one interim analysis. The target type-I error is $5\%.$ To determine the optimal values $(\beta_F,\beta_E)$, we used the same prior and utility functions of the previous simulation.  {\it In silico} trial replicates tailored to nGBM, have been generated following these two steps:

\begin{enumerate}
\item We sample with replacement a patient from the TMZ+RT  group of the  CENTRIC dataset. The patient is randomly assigned to either the TMZ+RT or the experimental arm of our {\it in silico} trial.
\item If the patient is assigned to our {\it in silico} TMZ+RT, we include the actual OS-24 and PFS-12 as primary and auxiliary outcomes. If the patient is assigned to the {\it in silico} experimental arm, we apply a simple, interpretable perturbation to the actual OS-24 and PFS-12 to produce scenarios with positive or negative treatment effects. If the actual outcome was negative, we relabel the individual primary outcome as positive with probability $p_y$. Similarly, we relabel a negative auxiliary outcome into a positive with probability $p_s$. The parameters $p_y$ and $p_s$ create {\it in silico} trials with treatment effects.  
\end{enumerate}

By adjusting the values of $p_{y}$ and $p_{s}$, we create scenarios to examine the trial design, with or without concordance of the treatment effects on primary and auxiliary outcomes. We considered three scenarios: 1. (null scenario) $p_{y} = p_{s} = 0$ there are no treatment effects on both primary and auxiliary outcome; 2. (concordant effects) $p_{y} = 0.4$ and $p_{s} = 0.5$; 3. (discordant effects) $p_{y} = 0$ and $p_{s} = 0.5,$ mimicking a trial where an effect is present for PFS but not on OS. We mention, as an example, the study of lomustine and bevacizumab~\citep{wick:2017},  in which the effects on primary and auxiliary outcomes were discordant. In our analysis, we compared the \textit{Auxiliary-Augmented}, \textit{Primary-Only}, and  \textit{Auxiliary-Only} designs. We did not model the individual enrollment times.

\begin{table}[b]
\center
\resizebox{\textwidth}{!}{ 
\begin{tabular}{r|cc|cc|cc} 
\toprule 
& \multicolumn{2}{c}{Scenario 1} & \multicolumn{2}{c}{Scenario 2} & \multicolumn{2}{c}{Scenario 3} \\ 
& \multicolumn{2}{c}{null scenario} & \multicolumn{2}{c}{concordant effects} & \multicolumn{2}{c}{discordant effects} \\ 
\midrule 
Method &  Power &   $\mathbb E[N]$ &
               Power &   $\mathbb E[N]$ &
               Power &   $\mathbb E[N]$  \\ 
Auxiliary-Augmented &0.050  (0.041; 0.009) & 113.8 &0.809  (0.556; 0.253) & 142.0 &0.052  (0.041; 0.011) & 151.6 \\ 
Primary-Only &0.051  (0.041; 0.010) & 127.9 &0.798  (0.556; 0.241) & 138.2 &0.051  (0.041; 0.010) & 128.2 \\ 
Auxiliary-Only &0.047  (0.037; 0.010) & 125.8 &0.988  (0.881; 0.107) & 111.4 &0.987  (0.877; 0.110) & 111.7 \\ 
\bottomrule 
\end{tabular} 

}
\caption{Proportion of times the null hypothesis $H_0$ of no treatment effect on the primary outcome (auxiliary outcome for the \textit{Auxiliary-Only} design)
is rejected (Power) and the expected sample size ($\mathbb E [N]$) for 5000 simulated trials using the CENTRIC data. In parentheses,  we report the frequency of simulations in which $H_0$ is rejected at the interim and final analysis, respectively.
}
\label{tab:sec_data}
\end{table}

 \tablename~\ref{tab:sec_data} summarizes the results.  In scenario 1, without treatment effects, all three designs control the type-I error at the $5\%.$ The \textit{Auxiliary-Augmented} design has the lowest average sample size with $114$ patients, compared to 128 for the \textit{Primary-Only} design, and $126$ for the  \textit{Auxiliary-Only} design. In scenario 2,  with treatment effects both on primary and auxiliary outcomes,  the \textit{Auxiliary-Only} design has an average sample size of 142 patients (81\% power)  compared to $138$ patients (80\% power) for the \textit{Primary-Only} design,  and $111$ (99\% power)  for the \textit{Auxiliary-Only} design. In scenario 3, with a treatment effect only on the auxiliary outcomes, as expected, the \textit{Auxiliary-Only} design has an inflated type-I error, while the  \textit{Auxiliary-Augmented} and \textit{Auxiliary-Only} designs control the type-I error at $5\%.$ Due to promising auxiliary data, the   \textit{Auxiliary-Augmented} design has an average sample size of  $152$ patients compared to  $128$ of the \textit{Auxiliary-Only} design.
 In \secname~\ref{sec:s3} of the Supplementary Material, we re-implement this retrospective analysis using time-to-event outcomes.

\section{Discussion}\label{sec:discussion}
Using early outcomes as OS surrogates across drug development can improve efficiency but also introduces risk, as discordant effects on early outcomes and OS are well documented \citep{merino:2023}. We propose a decision-making approach based on joint analyses of primary and auxiliary outcomes, enabling control of frequentist operating characteristics to meet regulatory or other stakeholder requirements. For example, in RCTs with multiple subgroups under FWER control, our approach can  increase efficiency compared with standard methods that ignore auxiliary data.

When auxiliary variables are used in interim and final analyses, it is important to balance the benefit of early efficacy signals against risks, such as planning a large phase III trial on promising auxiliary data without evidence of benefit on primary outcomes. We formalize this trade-off by maximizing explicit utility criteria that reflect the trial objectives, subject to frequentist constraints that control trial risks. 

Prior models are particularly attractive when real-world data and completed trials are available. If conflicting information exists about certain aspects of the prior, more flexible models can better represent this uncertainty. For example, our findings indicate that relaxing the assumption of concordant effects between primary and auxiliary outcomes causes only modest efficiency losses when effects are truly concordant, but provides greater robustness when effects diverge (\tablename~\ref{tab:s5}). In the presence of conflicting or weak prior information,  careful simulation-based evaluation of alternative prior specifications and designs is therefore essential. These results can also be used to revise the design and update the initial probability model.  However, when no information is available, using a vague or non-informative prior can aid analysis but may pose challenges for design, especially in determining sample size and other key design quantities. If reliable data on the relationship between the primary and auxiliary outcomes are unavailable, using auxiliary data in the trial's design and analysis may not be appropriate.

Our approach has several limitations. First, we use an approximation to the decision problem in~\eqref{eq:maximization}. However,  as discussed in \secname~\ref{sec:approximation}, the resulting parametric decision functions are simple, interpretable, and computationally attractive, as they restrict the design space to reasonable, easy-to-describe subsets. More complex parametric decision functions could also be explored,  for example,  the parametric decision function~\eqref{eq:procedure1w} could partition the $\alpha$ value based on a $\beta$ coefficient that in turn depends on the degree of correlation between effects on primary and auxiliary outcomes across subgroups, as indicated by the data. Second, frequentist control of false positives is guaranteed only asymptotically, so trial-specific simulations, such as those in \secname~\ref{sec:weighted_bonferroni}, remain important when designing a study. Finally, we have not systematically characterized the combinations of utility criteria and regulatory constraints under which auxiliary outcomes increase expected utility; instead, we provide illustrative examples where their use leads to higher expected utility.

The proposed framework can be applied to a range of decision problems, including complex designs with interim decisions, multiple groups, and multiple treatment arms. For example, in clinical studies that seek to identify patient subgroups with positive treatment effects, the large number of potential partitions raises concerns about cherry-picking~\citep{guo:2021}. In such settings, auxiliary and primary outcomes can be jointly used to select and recommend subgroups for subsequent confirmatory trials. An interpretable utility function can reflect the number of patients expected to benefit from treatment after regulatory approval, while frequentist constraints can bound the probability that a phase III confirmatory study fails to demonstrate efficacy. Thus, the framework leverages both auxiliary and primary outcomes to identify subgroups while controlling the risk of failure in a subsequent phase III trial.

\section*{Acknowledgments}
The examples in Section~\ref{sec:data_app}  and Supplementary S3 use information obtained from \url{www.projectdatasphere.org}, which is maintained by Project Data Sphere. Neither Project Data Sphere nor the owner(s) of any information from the website have contributed to, approved, or are in any way responsible for the contents of this publication. Lorenzo Trippa and Steffen Ventz have been supported by the NIH grant 	R01LM013352.

\newcommand{\beginsupplement}{%
        \setcounter{table}{0}
        \renewcommand{\thetable}{S\arabic{table}}%
        \setcounter{figure}{0}
        \renewcommand{\thefigure}{S\arabic{figure}}%
	\setcounter{section}{0}
        \renewcommand{\thesection}{S\arabic{section}}
     }

\section*{Supplementary Material}
\beginsupplement
\section{Proof of the propositions}
\textbf{Proposition 3.1} The set $\argmax _{\varphi \in \varPhi_{\alpha,h} }
 \mathbb E^\pi[ \varphi(\Data) ]$
 contains at least one decision function in $\varPhi'_{\alpha,h}$ 

\begin{proof} 
 Let $\varphi_{Y,S}(y,s)$ be a solution of (3). This function maps the data $(Y,S)$ into the interval $[0,1]$. We define $\varphi_{Y}(y) = \int_{\mathcal S} \varphi_{Y,S}(y,s) \pi(s \mid y) \d s$, where $\pi (s \mid y )$ is the conditional distribution of 
 $S$ given $Y$. The function $\varphi_{Y}$ maps the primary information $Y$ into the interval $[0,1]$ and does not depend on $S$. Moreover, 
\begin{eqnarray*}
\mathbb E^\pi \left[ \varphi_{Y}(Y) \right] &=&
\int_{\mathcal Y }
\int_\Theta
\varphi_{Y}(y) 
p_\theta(y) \pi(\theta)\d\theta \d y= \mathbb E^\pi \left[ \varphi_{Y,S}(Y,S) \right].
\end{eqnarray*}
It remains to show that $\varphi(Y)$ controls the type I error rate at level $\alpha$. 
For a general $\tilde \theta \in H_0$ 
\begin{eqnarray*}
 \mathbb E_{\tilde \theta} \left[ \varphi_{Y}(Y) \right] =
 \mathbb E_{\tilde \theta} \left[\mathbb E^\pi [\varphi_{Y,S}(Y,S) \mid Y] \right] &=& \int_{\mathcal Y} \left [
\int_{\mathcal S}
\varphi_{Y,S}(y,s) \pi(s\mid y) \d s \right] p_{\tilde \theta}(y) \d y \\
&=&
\int_{\mathcal S}\int_{\mathcal Y} \varphi_{Y,S}(y,s) \pi(s\mid y) p_{\tilde \theta}(y) \d y \d s \leq \alpha,
\end{eqnarray*}
where $p_{\tilde \theta}(y)$ is the distribution of $Y$ when $\theta=\tilde\theta$. The last inequality follows from the fact that the product $\pi(s\mid y) p_{\tilde \theta}(y)$ identifies a joint distribution for $(S,Y)$ in $H_0$ and $ \varphi_{Y,S}(y,s) \in \varPhi'_{\alpha,h}$. Note that $\Theta$ and $H_0$ were defined without assuming the independence of primary and auxiliary outcomes across patients.
\end{proof}

\textbf{Proposition 5.1}
Assume that for each $\btheta \in \Theta$ the means $\mathbb E_\theta( \vert S_{k,i} \vert \mid C_{k,i}=c) $ and variances $\mbox{Var}_\theta(Y_{k,i} \mid C_{k,i}=c) $ are finite for all $k=1, \ldots, K$ and $c=0,1$. 
Consider a sequence of clinical trials with increasing sample sizes such that $N_k \rightarrow \infty$ for $k=1, \ldots, K$. Then, for any $\bbeta \in \mathbb R^K$, the decision function $\varphi_{\bbeta} : \mathcal{ D} \rightarrow \{0,1\}^K$ 
 in (6)  with weights defined in (7)  controls asymptotically the FWER at level $\alpha$,
$$\lim_{N\rightarrow \infty} \mathbb E_{\btheta}[h(\varphi_{\beta}(\Data), \btheta)] \leq \alpha,\;\;\;\;\;\;\;\;\;\;\;\;\; \forall \theta\in\Theta,$$
 where $h(\varphi_\beta(\Data), \btheta ) = 1\left \{ \sum_{k=1}^K \varphi_{\beta,k}(\Data) 1 \{ \gamma_k \leq 0\} \geq 1 \right \}$.
\begin{proof}
	 By law of large numbers, the weights in~(7) converge to $\boldsymbol \omega^{\infty}(\bbeta) = (\omega^\infty_1(\bbeta) ,\ldots, \omega^\infty_K(\bbeta))$, 
	 where $\omega^\infty_k(\bbeta) \propto \exp\{ \beta_k ( \mathbb E_{\theta}[ S_{k,1} \mid C_{k,1}=1] - \mathbb E_{\theta}[ S_{k,1} \mid C_{k,1}=0] )\}.$ Using Markov's inequality, we have that
\begin{align*}
 p_\theta\left(\sum_{k : \gamma_k\le 0}
1\{\text{\tt pv}_{k} \leq \alpha w^\infty_k(\bbeta)\} \geq1 \right) 
 &\leq \mathbb{E}_\theta\left(\sum_{k : \gamma_k\le 0 } 
 1\{\text{\tt pv}_{k} \leq \alpha w^\infty_k(\bbeta)\}
 \right)
 \\
 & \leq \sum_{k : \gamma_k\le 0 } \mathbb{E}_\theta\left(1\{ \text{\tt pv}_{k} \leq \alpha w^\infty_k(\bbeta)\} \right) \le \alpha. 
\end{align*} 
The asymptotic control of the FWER at level $\alpha$ now follows from 
the facts that (i) as $N$ increases $\varphi_{\bbeta}(\Data)$ converges almost surely to
 $\Big( 1\{\text{\tt pv}_k \leq \alpha w^\infty_k(\bbeta) \}, k=1, \ldots K \Big) $, and 
 (ii) $\varphi_{\bbeta}(\Data)$ includes $K$ bounded and hence uniformly integrable functions.
\end{proof}

\section{Supplementary figures and tables}
Code to reproduce figures and tables is available at \url{https://github.com/rMassimiliano/primary\_and\_auxiliary}.

\begin{table}[h!]
\center
\begin{tabular}{lrrrrrr}
\toprule
  & Mean & Min. & 1st Qu. & Median &  3rd Qu. & Max.\\
\midrule
Proportion Y=1 (SOC) & 0.238 & 0.029 & 0.175 & 0.230 &  0.294 & 0.657\\
Proportions Y=1 (Treated) & 0.242 & 0.017 & 0.176 & 0.231 &  0.299 & 0.76
2\\
Difference in proportions (TE) for Y & 0.003 & -0.236 & -0.038 & 0.002 & 
 0.044 & 0.289\\
Proportion S=1 (SOC) & 0.349 & 0.069 & 0.276 & 0.344 &  0.416 & 0.764\\
Proportions S=1 (Treated) & 0.352 & 0.039 & 0.274 & 0.347 &  0.424 & 0.81
4\\
Difference in proportions (TE) for S & 0.003 & -0.380 & -0.046 & 0.002 & 
 0.050 & 0.342\\
Correlation between Y and S & 0.138 & -0.197 & 0.085 & 0.138  &0.193 & 0
.414\\
\bottomrule
\end{tabular}
\caption{Summaries of some characteristics of the prior model used in Section 5  (simulation experiments) of the main manuscript.  Summaries have been calculated simulating $5000$ trials with a sample size of $200$ and prevalence of $0.6$ for subgroup one and $0.4$ for subgroup two. We assumed the same prior for the subgroups $k=1,2$ with
$\zeta^S_{0,k} \sim \mathcal N(-0.8,0.5^2),$
$\zeta^Y_{0,k} \sim \mathcal N(-1.5,0.5^2),$
$\sigma^2_k = 1,$
$\xi =0.1,$
$m_{S,k} = 0,$
$\sigma^2_{S,k} = 0.8,$
$v_{k} = 6$ , 
$o_{k} =1$. The correlation between the treatment effect (difference in proportions) measured with primary and auxiliary outcomes is approximately $0.29.$
}
\label{tab:s1}
\end{table}

\begin{figure}[h!]
\begin{center}
 \includegraphics[width = 0.60\textwidth]{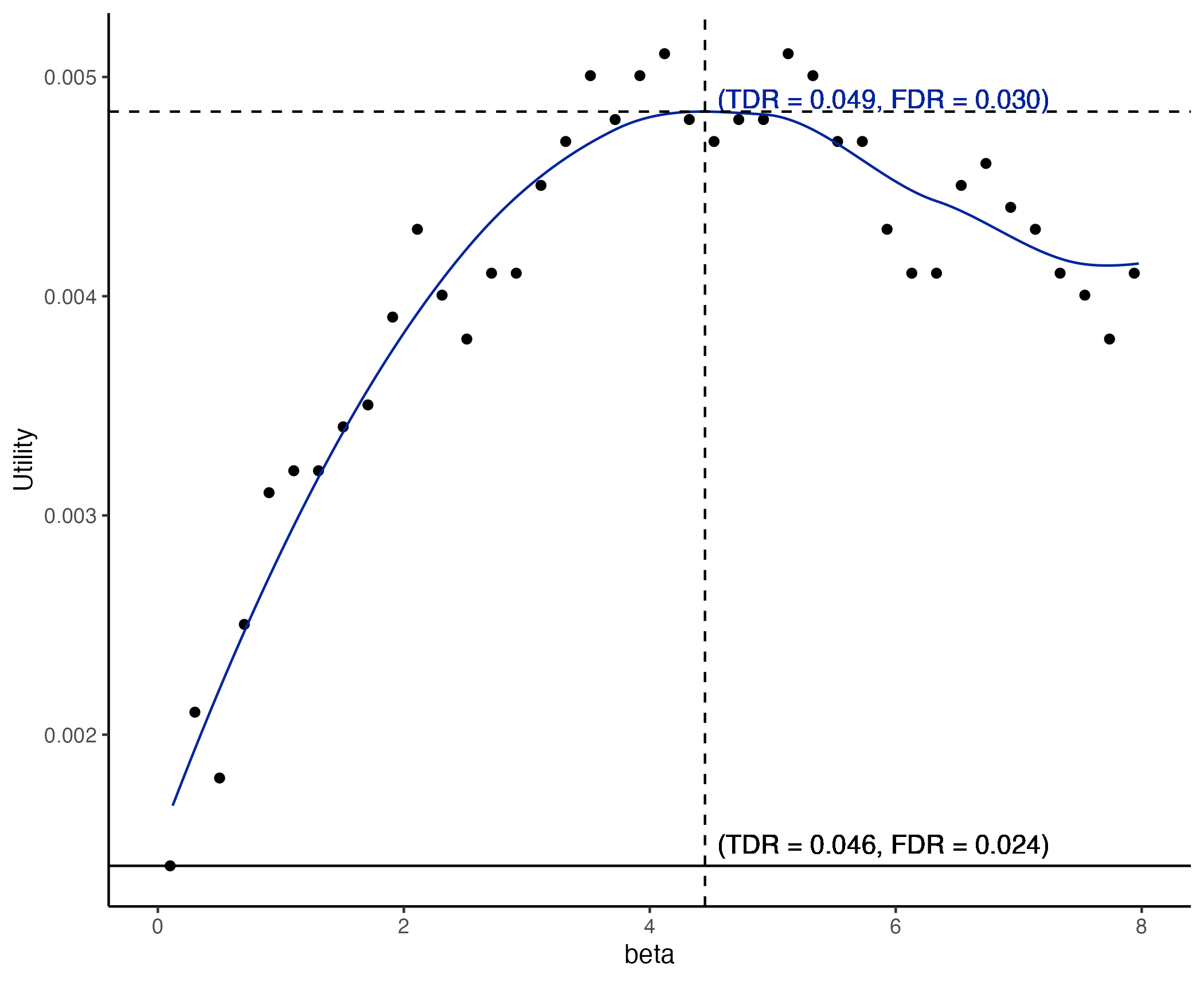}
\end{center}
\caption{Expected utility for the simulation study presented in \secname~{5.1}, for $K=2$ subgroups, and $\lambda = 0.5$. The expected utility has been computed on a grid of $\beta$ values with 5000 trial replicates from the prior model.  The solid blue line is a smooth estimate of the expected utility obtained using a Local Polynomial Regression Fitting (loess) (R function \texttt{loess} with \texttt{span} = 0.4). The dashed lines indicate the optimal solution $\beta = 4.45$ (vertical line)  and its expected utility (horizontal line), while the horizontal solid line indicates the expected utility associated with the Bonferroni method, corresponding to $\beta=0$. We also report the average False Discovery Rate (FDR), defined as the number of false rejections divided by the total number of rejections, and the True Discovery Rate (TDR), defined as the number of true rejections by the total number of rejections, for the optimal {\it Auxiliary-Augmented} method, and the Bonferroni method. The average FDR and TDR are also computed using 5000 replicates from the prior model. }
\label{fig:s1}
\end{figure}

\begin{table}[h!]
\centering
\begin{tabular}{rccc|ccc} 
\toprule 
& \multicolumn{3}{c}{Scenario 1} & \multicolumn{3}{c}{Scenario 2} \\ 
\midrule 
Method  & $ R_{k,c}= 1$ & $R_{k,c} = 2$ & $R_{k,c} = 10$  & $ R_{k,c}= 1$ & $R_{k,c} = 2$ & $R_{k,c} = 10$ \\ 
Auxiliary-Augmented & 0.052 & 0.059 &0.057 &0.052 &0.053 &0.059 \\ 
Auxiliary-Augmented-B & 0.050 & 0.053 &0.048 &0.053 &0.050 &0.051 \\ 
Bonferroni & 0.051 & 0.056 &0.049 &0.054 &0.050 &0.053 \\ 
Holm & 0.051 & 0.056 &0.049 &0.054 &0.050 &0.053 \\ 
FAB & 0.042 & 0.046 &0.042 &0.043 &0.044 &0.044 \\ 
Auxiliary-Only & 0.052 & 0.051 &0.054 &0.825 &0.825 &0.831 \\ 
\bottomrule 
\end{tabular}
\\
\begin{tabular}{rccc|ccc} 
 & \multicolumn{6}{c}{$K=6$}  \\ 
\toprule 
& \multicolumn{3}{c}{Scenario 1} & \multicolumn{3}{c}{Scenario 2} \\ 
\midrule 
Method  & $ R_{k,c}= 1$ & $R_{k,c} = 2$ & $R_{k,c} = 10$  & $ R_{k,c}= 1$ & $R_{k,c} = 2$ & $R_{k,c} = 10$ \\ 
Auxilary-Augmented & 0.056 & 0.076 &0.090 &0.057 &0.061 &0.082 \\ 
Auxilary-Augmented-B & 0.055 & 0.058 &0.047 &0.052 &0.047 &0.045 \\ 
Bonferroni & 0.053 & 0.060 &0.057 &0.055 &0.053 &0.060 \\ 
Holm & 0.053 & 0.060 &0.057 &0.055 &0.053 &0.060 \\ 
FAB & 0.030 & 0.036 &0.033 &0.033 &0.030 &0.036 \\ 
Auxiliary-Only & 0.061 & 0.058 &0.057 &0.809 &0.811 &0.801 \\ 
\bottomrule 
\end{tabular}
\caption{Family Wise Error Rate (FWER) for the simulations presented in Section 5 of the main manuscript.}
\label{tab:s2}
\end{table}

\begin{table}[h!]
{\centering
 \resizebox{0.9\textwidth}{!}{
\begin{tabular}{rcccc>{\columncolor[gray]{0.9}}cccccc}
\toprule 

\multicolumn{10}{c}{Scenario 1 $R_{k,c} = 1$}\\ 
\midrule 
($\lambda, \beta_{YS})$ & (0.1, 14.07)&(0.2, 13.91)&(0.3, 13.75)&(0.4, 13.43)&(0.5, 4.45)&(0.6, 4.21)&(0.7, 4.05)&(0.8, 3.93)&(0.9, 3.81)&(1.0, 3.69) \\ 
subgroup 1 &0.022 &0.022 &0.022 &0.022 &0.023 &0.023 &0.023 &0.023 &0.023 &0.023 \\ 
subgroup 2 &0.029 &0.029 &0.029 &0.029 &0.029 &0.029 &0.029 &0.029 &0.029 &0.029 \\ 
\bottomrule  
\multicolumn{10}{c}{Scenario 1 $R_{k,c} = 2$}\\ 
\midrule 
($\lambda, \beta_{YS})$ & (0.1, 14.07)&(0.2, 13.91)&(0.3, 13.75)&(0.4, 13.43)&(0.5, 4.45)&(0.6, 4.21)&(0.7, 4.05)&(0.8, 3.93)&(0.9, 3.81)&(1.0, 3.69) \\ 
subgroup 1 &0.030 &0.030 &0.030 &0.030 &0.030 &0.030 &0.030 &0.030 &0.030 &0.030 \\ 
subgroup 2 &0.031 &0.031 &0.031 &0.031 &0.031 &0.030 &0.030 &0.030 &0.030 &0.030 \\ 
\bottomrule  
\multicolumn{10}{c}{Scenario 1 $R_{k,c} = 10$}\\ 
\midrule 
($\lambda, \beta_{YS})$ & (0.1, 14.07)&(0.2, 13.91)&(0.3, 13.75)&(0.4, 13.43)&(0.5, 4.45)&(0.6, 4.21)&(0.7, 4.05)&(0.8, 3.93)&(0.9, 3.81)&(1.0, 3.69) \\ 
subgroup 1 &0.032 &0.032 &0.032 &0.032 &0.030 &0.030 &0.030 &0.030 &0.030 &0.030 \\ 
subgroup 2 &0.033 &0.033 &0.033 &0.033 &0.028 &0.028 &0.028 &0.028 &0.028 &0.028 \\ 
\bottomrule  
\multicolumn{10}{c}{Scenario 2 $R_{k,c} = 1$}\\ 
\midrule 
($\lambda, \beta_{YS})$ & (0.1, 14.07)&(0.2, 13.91)&(0.3, 13.75)&(0.4, 13.43)&(0.5, 4.45)&(0.6, 4.21)&(0.7, 4.05)&(0.8, 3.93)&(0.9, 3.81)&(1.0, 3.69) \\ 
subgroup 1 &0.047 &0.047 &0.047 &0.046 &0.039 &0.039 &0.038 &0.038 &0.038 &0.037 \\ 
subgroup 2 &0.004 &0.005 &0.005 &0.005 &0.014 &0.015 &0.015 &0.016 &0.016 &0.016 \\ 
\bottomrule  
\multicolumn{10}{c}{Scenario 2 $R_{k,c} = 2$}\\ 
\midrule 
($\lambda, \beta_{YS})$ & (0.1, 14.07)&(0.2, 13.91)&(0.3, 13.75)&(0.4, 13.43)&(0.5, 4.45)&(0.6, 4.21)&(0.7, 4.05)&(0.8, 3.93)&(0.9, 3.81)&(1.0, 3.69) \\ 
subgroup 1 &0.044 &0.043 &0.043 &0.043 &0.036 &0.036 &0.036 &0.035 &0.035 &0.035 \\ 
subgroup 2 &0.008 &0.008 &0.008 &0.009 &0.017 &0.017 &0.018 &0.018 &0.018 &0.018 \\ 
\bottomrule  
\multicolumn{10}{c}{Scenario 2 $R_{k,c} = 10$}\\ 
\midrule 
($\lambda, \beta_{YS})$ & (0.1, 14.07)&(0.2, 13.91)&(0.3, 13.75)&(0.4, 13.43)&(0.5, 4.45)&(0.6, 4.21)&(0.7, 4.05)&(0.8, 3.93)&(0.9, 3.81)&(1.0, 3.69) \\ 
subgroup 1 &0.050 &0.050 &0.050 &0.050 &0.041 &0.040 &0.039 &0.039 &0.039 &0.038 \\ 
subgroup 2 &0.010 &0.010 &0.010 &0.010 &0.019 &0.020 &0.020 &0.020 &0.020 &0.020 \\ 
\bottomrule  
\multicolumn{10}{c}{Scenario 3 $R_{k,c} = 1$}\\ 
\midrule 
($\lambda, \beta_{YS})$ & (0.1, 14.07)&(0.2, 13.91)&(0.3, 13.75)&(0.4, 13.43)&(0.5, 4.45)&(0.6, 4.21)&(0.7, 4.05)&(0.8, 3.93)&(0.9, 3.81)&(1.0, 3.69) \\ 
subgroup 1 &0.625 &0.627 &0.627 &0.628 &0.673 &0.673 &0.673 &0.673 &0.673 &0.674 \\ 
subgroup 2 &0.027 &0.027 &0.027 &0.027 &0.026 &0.027 &0.027 &0.027 &0.027 &0.027 \\ 
\bottomrule  
\multicolumn{10}{c}{Scenario 3 $R_{k,c} = 2$}\\ 
\midrule 
($\lambda, \beta_{YS})$ & (0.1, 14.07)&(0.2, 13.91)&(0.3, 13.75)&(0.4, 13.43)&(0.5, 4.45)&(0.6, 4.21)&(0.7, 4.05)&(0.8, 3.93)&(0.9, 3.81)&(1.0, 3.69) \\ 
subgroup 1 &0.631 &0.631 &0.631 &0.635 &0.673 &0.673 &0.673 &0.674 &0.674 &0.674 \\ 
subgroup 2 &0.030 &0.030 &0.030 &0.030 &0.031 &0.031 &0.031 &0.031 &0.030 &0.031 \\ 
\bottomrule  
\multicolumn{10}{c}{Scenario 3 $R_{k,c} = 10$}\\ 
\midrule 
($\lambda, \beta_{YS})$ & (0.1, 14.07)&(0.2, 13.91)&(0.3, 13.75)&(0.4, 13.43)&(0.5, 4.45)&(0.6, 4.21)&(0.7, 4.05)&(0.8, 3.93)&(0.9, 3.81)&(1.0, 3.69) \\ 
subgroup 1 &0.618 &0.619 &0.620 &0.621 &0.666 &0.666 &0.667 &0.667 &0.668 &0.669 \\ 
subgroup 2 &0.042 &0.042 &0.041 &0.042 &0.036 &0.036 &0.036 &0.035 &0.035 &0.035 \\ 
\bottomrule  
\multicolumn{10}{c}{Scenario 4 $R_{k,c} = 1$}\\ 
\midrule 
($\lambda, \beta_{YS})$ & (0.1, 14.07)&(0.2, 13.91)&(0.3, 13.75)&(0.4, 13.43)&(0.5, 4.45)&(0.6, 4.21)&(0.7, 4.05)&(0.8, 3.93)&(0.9, 3.81)&(1.0, 3.69) \\ 
subgroup 1 &0.757 &0.757 &0.757 &0.757 &0.731 &0.729 &0.727 &0.727 &0.726 &0.725 \\ 
subgroup 2 &0.006 &0.006 &0.006 &0.006 &0.014 &0.013 &0.013 &0.013 &0.014 &0.014 \\ 
\bottomrule  
\multicolumn{10}{c}{Scenario 4 $R_{k,c} = 2$}\\ 
\midrule 
($\lambda, \beta_{YS})$ & (0.1, 14.07)&(0.2, 13.91)&(0.3, 13.75)&(0.4, 13.43)&(0.5, 4.45)&(0.6, 4.21)&(0.7, 4.05)&(0.8, 3.93)&(0.9, 3.81)&(1.0, 3.69) \\ 
subgroup 1 &0.760 &0.759 &0.759 &0.759 &0.735 &0.733 &0.732 &0.731 &0.728 &0.728 \\ 
subgroup 2 &0.008 &0.008 &0.008 &0.009 &0.018 &0.019 &0.020 &0.021 &0.021 &0.021 \\ 
\bottomrule  
\multicolumn{10}{c}{Scenario 4 $R_{k,c} = 10$}\\ 
\midrule 
($\lambda, \beta_{YS})$ & (0.1, 14.07)&(0.2, 13.91)&(0.3, 13.75)&(0.4, 13.43)&(0.5, 4.45)&(0.6, 4.21)&(0.7, 4.05)&(0.8, 3.93)&(0.9, 3.81)&(1.0, 3.69) \\ 
subgroup 1 &0.763 &0.763 &0.763 &0.763 &0.731 &0.729 &0.728 &0.727 &0.726 &0.724 \\ 
subgroup 2 &0.011 &0.011 &0.011 &0.011 &0.016 &0.016 &0.016 &0.017 &0.017 &0.018 \\ 
\bottomrule  
\multicolumn{10}{c}{Scenario 5 $R_{k,c} = 1$}\\ 
\midrule 
($\lambda, \beta_{YS})$ & (0.1, 14.07)&(0.2, 13.91)&(0.3, 13.75)&(0.4, 13.43)&(0.5, 4.45)&(0.6, 4.21)&(0.7, 4.05)&(0.8, 3.93)&(0.9, 3.81)&(1.0, 3.69) \\ 
subgroup 1 &0.338 &0.341 &0.344 &0.352 &0.578 &0.586 &0.590 &0.592 &0.595 &0.598 \\ 
subgroup 2 &0.044 &0.044 &0.044 &0.044 &0.037 &0.036 &0.036 &0.036 &0.036 &0.036 \\ 
\bottomrule  
\multicolumn{10}{c}{Scenario 5 $R_{k,c} = 2$}\\ 
\midrule 
($\lambda, \beta_{YS})$ & (0.1, 14.07)&(0.2, 13.91)&(0.3, 13.75)&(0.4, 13.43)&(0.5, 4.45)&(0.6, 4.21)&(0.7, 4.05)&(0.8, 3.93)&(0.9, 3.81)&(1.0, 3.69) \\ 
subgroup 1 &0.336 &0.339 &0.343 &0.349 &0.564 &0.572 &0.577 &0.581 &0.584 &0.587 \\ 
subgroup 2 &0.050 &0.050 &0.050 &0.050 &0.041 &0.041 &0.040 &0.040 &0.040 &0.039 \\ 
\bottomrule  
\multicolumn{10}{c}{Scenario 5 $R_{k,c} = 10$}\\ 
\midrule 
($\lambda, \beta_{YS})$ & (0.1, 14.07)&(0.2, 13.91)&(0.3, 13.75)&(0.4, 13.43)&(0.5, 4.45)&(0.6, 4.21)&(0.7, 4.05)&(0.8, 3.93)&(0.9, 3.81)&(1.0, 3.69) \\ 
subgroup 1 &0.337 &0.339 &0.341 &0.349 &0.568 &0.574 &0.578 &0.582 &0.585 &0.588 \\ 
subgroup 2 &0.050 &0.050 &0.050 &0.050 &0.043 &0.042 &0.042 &0.042 &0.042 &0.042 \\ 
\bottomrule  
\end{tabular}

 }}
 \caption{Sensitivity analyses for the $\lambda$ parameter in the utility function (5). For the simulation study presented in Section 5.1, for K = 2 subgroups, we compute the proportion of times the null hypotheses $H_{0,k}$ have been rejected for subgroup $k=1,2$ for ten {\it Auxiliary-Augmented} procedures across 5000 trials. We consider $\lambda \in \{0.1,0.2,0.3,\ldots,0.9,1\}.$ The $\lambda$ used in the simulations study in the main manuscript, $\lambda = 0.5,$ is highlighted in grey.}
 \label{tab:s3}
\end{table}

\clearpage

\begin{figure}[h!]
\begin{center}
\includegraphics[width = 0.69\textwidth]{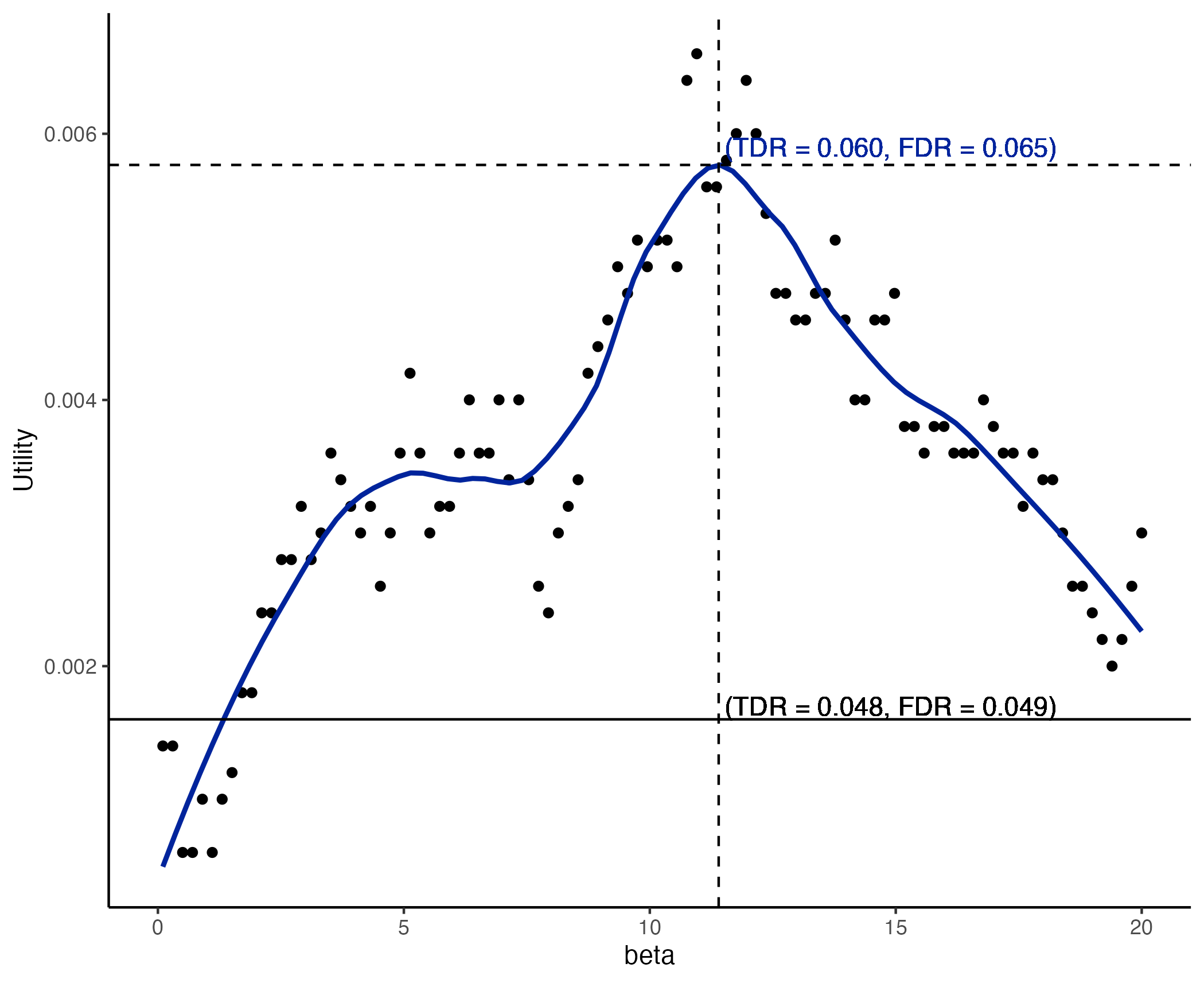}
\end{center}
\caption{Expected utility for the simulation study presented in \secname~{5.1}, for $K=6$ subgroups, and $\lambda = 1.0$. The expected utility has been computed on a grid of $\beta$ values with 5000 trial replicates.  The solid blue line is a smooth estimate of the expected utility obtained using a Local Polynomial Regression Fitting (loess) (R function \texttt{loess} with \texttt{span} = 0.4). The dashed lines indicate the optimal solution $\beta = 11.4$ (vertical line)  and its expected utility (horizontal line), while the horizontal solid line indicates the expected utility associated with the Bonferroni method. We also report the average False Discovery Rate (FDR), defined as the number of false rejections divided by the total number of rejections, and the True Discovery Rate (TDR), defined as the number of true rejections by the total number of rejections, for the optimal {\it Auxiliary-Augmented} method, and the Bonferroni method. The average FDR and TDR are also computed using 5000 replicates from the prior model.}
\label{fig:s2}
\end{figure}

\begin{table}[h!]
\centering
 \resizebox{0.9\textwidth}{!}{\begin{tabular}{rcccccc} 
\toprule 
\multicolumn{7}{c}{Scenario 1}\\ 
\midrule 
& \multicolumn{2}{c}{$ R_{k,c}= 1$} & \multicolumn{2}{c}{$R_{k,c} = 2$} & \multicolumn{2}{c}{$R_{k,c} = 10$} \\ 
Method                & subgroup 1 & subgroup 2--6 & subgroup 1 & subgroup 2--6 & subgroup 1 & subgroup 2--6 \\
Auxiliary-Augmented   & 0.009 & 0.048 &0.010 & 0.066 &0.012 & 0.079 \\
Auxiliary-Augmented-B & 0.009 & 0.047 &0.008 & 0.051 &0.006 & 0.041 \\
Bonferroni            & 0.010 & 0.044 &0.009 & 0.051 &0.009 & 0.048 \\
Holm                  & 0.010 & 0.044 &0.009 & 0.051 &0.009 & 0.048 \\
FAB                   & 0.007 & 0.023 &0.007 & 0.029 &0.007 & 0.026 \\
Auxiliary-Only        & 0.013 & 0.048 &0.009 & 0.050 &0.010 & 0.047 \\
\bottomrule 
\multicolumn{7}{c}{Scenario 2}\\ 
\midrule 
& \multicolumn{2}{c}{$ R_{k,c} = 1$} & \multicolumn{2}{c}{$ R_{k,c} = 2$} & \multicolumn{2}{c}{$ R_{k,c} = 10$} \\Method  & subgroup 1 & subgroup 2--6 & subgroup 1 & subgroup 2--6 & subgroup 1 & subgroup 2--6 \\ 
Auxiliary-Augmented   & 0.035 & 0.022 &0.033 & 0.028 &0.039 & 0.043 \\
Auxiliary-Augmented-B & 0.033 & 0.020 &0.026 & 0.022 &0.022 & 0.023 \\
Bonferroni            & 0.009 & 0.047 &0.007 & 0.046 &0.011 & 0.049 \\
Holm                  & 0.009 & 0.047 &0.007 & 0.046 &0.011 & 0.049 \\
FAB                   & 0.007 & 0.027 &0.005 & 0.025 &0.008 & 0.028 \\
Auxiliary-Only        & 0.799 & 0.046 &0.801 & 0.050 &0.791 & 0.047 \\
\bottomrule 
\multicolumn{7}{c}{Scenario 3}\\ 
\midrule 
& \multicolumn{2}{c}{$ R_{k,c} = 1$} & \multicolumn{2}{c}{$R_{k,c} = 2$} & \multicolumn{2}{c}{$R_{k,c} = 10$} \\ 
Method  & subgroup 1 & subgroup 2--6 & subgroup 1 & subgroup 2--6 & subgroup 1 & subgroup 2--6 \\ 
Auxiliary-Augmented    & 0.560 & 0.047 &0.565 & 0.061 &0.560 & 0.085 \\
 Auxiliary-Augmented-B & 0.555 & 0.046 &0.528 & 0.045 &0.456 & 0.037 \\
 Bonferroni            & 0.625 & 0.049 &0.637 & 0.047 &0.629 & 0.046 \\
Holm                   & 0.626 & 0.054 &0.638 & 0.054 &0.629 & 0.053 \\
FAB                    & 0.597 & 0.030 &0.613 & 0.027 &0.603 & 0.027 \\
Auxiliary-Only         & 0.008 & 0.048 &0.011 & 0.044 &0.009 & 0.051 \\
\bottomrule 
\multicolumn{7}{c}{Scenario 4}\\ 
\midrule 
& \multicolumn{2}{c}{$R_{k,c} = 1$} & \multicolumn{2}{c}{$R_{k,c} = 2$} & \multicolumn{2}{c}{$R_{k,c} = 10$} \\ 
Method  & subgroup 1 & subgroup 2--6 & subgroup 1 & subgroup 2--6 & subgroup 1 & subgroup 2--6 \\ 
Auxiliary-Augmented    & 0.799 & 0.019 &0.791 & 0.032 &0.780 & 0.046 \\
 Auxiliary-Augmented-B & 0.792 & 0.019 &0.752 & 0.024 &0.688 & 0.024 \\
 Bonferroni            & 0.629 & 0.047 &0.632 & 0.052 &0.630 & 0.053 \\
Holm                   & 0.631 & 0.050 &0.633 & 0.057 &0.630 & 0.061 \\
FAB                    & 0.602 & 0.025 &0.609 & 0.028 &0.607 & 0.028 \\
Auxiliary-Only         & 0.795 & 0.045 &0.799 & 0.045 &0.798 & 0.046 \\
\bottomrule\multicolumn{7}{c}{Scenario 5}\\ 
\midrule 
& \multicolumn{2}{c}{$R_{k,c} = 1$} & \multicolumn{2}{c}{$R_{k,c} = 2$} & \multicolumn{2}{c}{$R_{k,c} = 10$} \\ 
Method  & subgroup 1 & subgroup 2--6 & subgroup 1 & subgroup 2--6 & subgroup 1 & subgroup 2--6 \\ 
Auxiliary-Augmented    & 0.276 & 0.058 &0.277 & 0.067 &0.284 & 0.089 \\
 Auxiliary-Augmented-B & 0.271 & 0.055 &0.251 & 0.053 &0.218 & 0.048 \\
 Bonferroni            & 0.631 & 0.049 &0.619 & 0.045 &0.621 & 0.048 \\
Holm                   & 0.633 & 0.055 &0.619 & 0.051 &0.623 & 0.052 \\
FAB                    & 0.610 & 0.027 &0.591 & 0.023 &0.594 & 0.026 \\
Auxiliary-Only         & 0.000 & 0.050 &0.000 & 0.045 &0.000 & 0.051 \\
\bottomrule 
\end{tabular}
}
 \caption{Proportion of time the null hypothesis $H_{0,1}$ has been rejected in the first subgroup, and proportion of time at least one of the null $H_{0,k}$, for $k=1,\ldots,6,$ across 5000 trials, for the methods described in \secname~{5} of the main manuscript. The Auxiliary-Augmented-B method refers to the bootstrap calibrated procedure described in Section~{5.2}.}
 \label{tab:sim_k6}
\end{table}

\begin{figure}[h!]
\centering
\includegraphics[width = 0.5\textwidth]{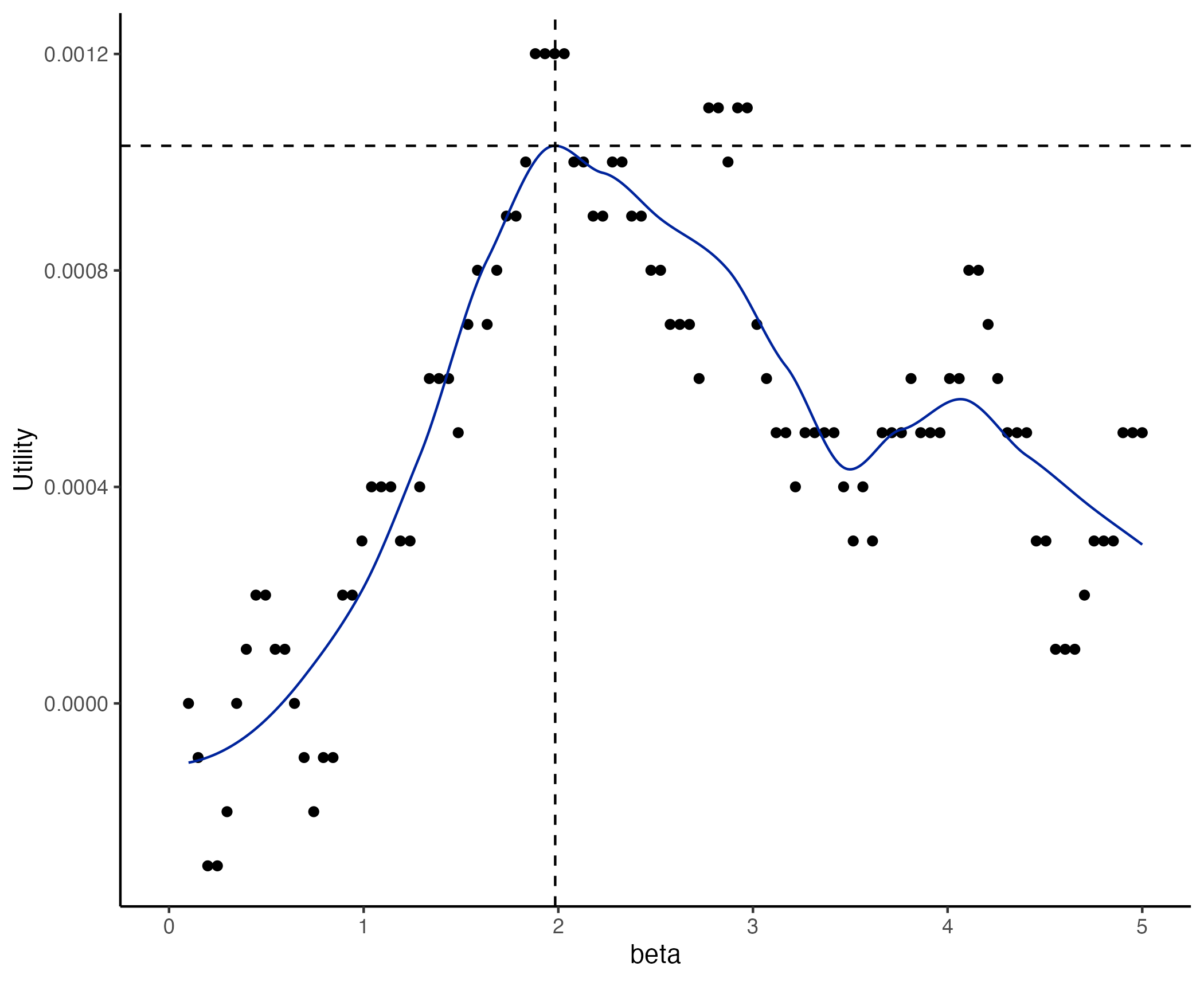}
\caption{
Expected utility with a variation of the prior model described in \secname~{5} and $\lambda = 0.5$. 
Compared to \secname~{5}, the prior changes the distribution of  $c_{Y,k},$ for $k=1,2,$ to a mixture of two components, a beta distribution with parameters $v_k = 6$ and $o_k=1$ with probability $0.9$ and a point mass at zero with probability $0.1$.  In this case, the optimal design is parametrized by $\beta_{Y,S} = 1.98.$
}
\label{fig:s3}
\end{figure}
\begin{table}[h!]
\centering
\resizebox{0.5\textwidth}{!}{
\begin{tabular}{r|cc|cc|cc}
\toprule 
& \multicolumn{2}{c}{$R_{k,c} = 1$} & \multicolumn{2}{c}{$R_{k,c} = 2$} & \multicolumn{2}{c}{$R_{k,c} = 10$} \\
Scenario 1 & 0.023 &0.027 &0.029 &0.028 &0.028 &0.026 \\
Scenario 2 & 0.034 &0.020 &0.031 &0.022 &0.033 &0.022 \\
Scenario 3 & 0.677 &0.027 &0.679 &0.030 &0.677 &0.034 \\
Scenario 4 & 0.711 &0.017 &0.706 &0.026 &0.706 &0.019 \\
Scenario 5 & 0.643 &0.032 &0.631 &0.035 &0.634 &0.037 \\
\bottomrule  
\end{tabular}

}
\caption{Operative characteristics for the {\it Auxiliary-Augmented} procedure with a variation of the prior model described in \secname~{5}. Compared to \secname~{5}, the prior changes the distribution of  $c_{Y,k},$ for $k=1,2,$ to a mixture of two components, a beta distribution with parameters $v_k = 6$ and $o_k=1$ with probability $0.9$ and a point mass at zero with probability $0.1$.  In this case, the optimal design is parametrized by $\beta_{Y,S} = 1.98$.
}
\label{tab:s5}
\end{table}

\begin{figure}[h!]
\begin{center}
\includegraphics[width = 0.69\textwidth]{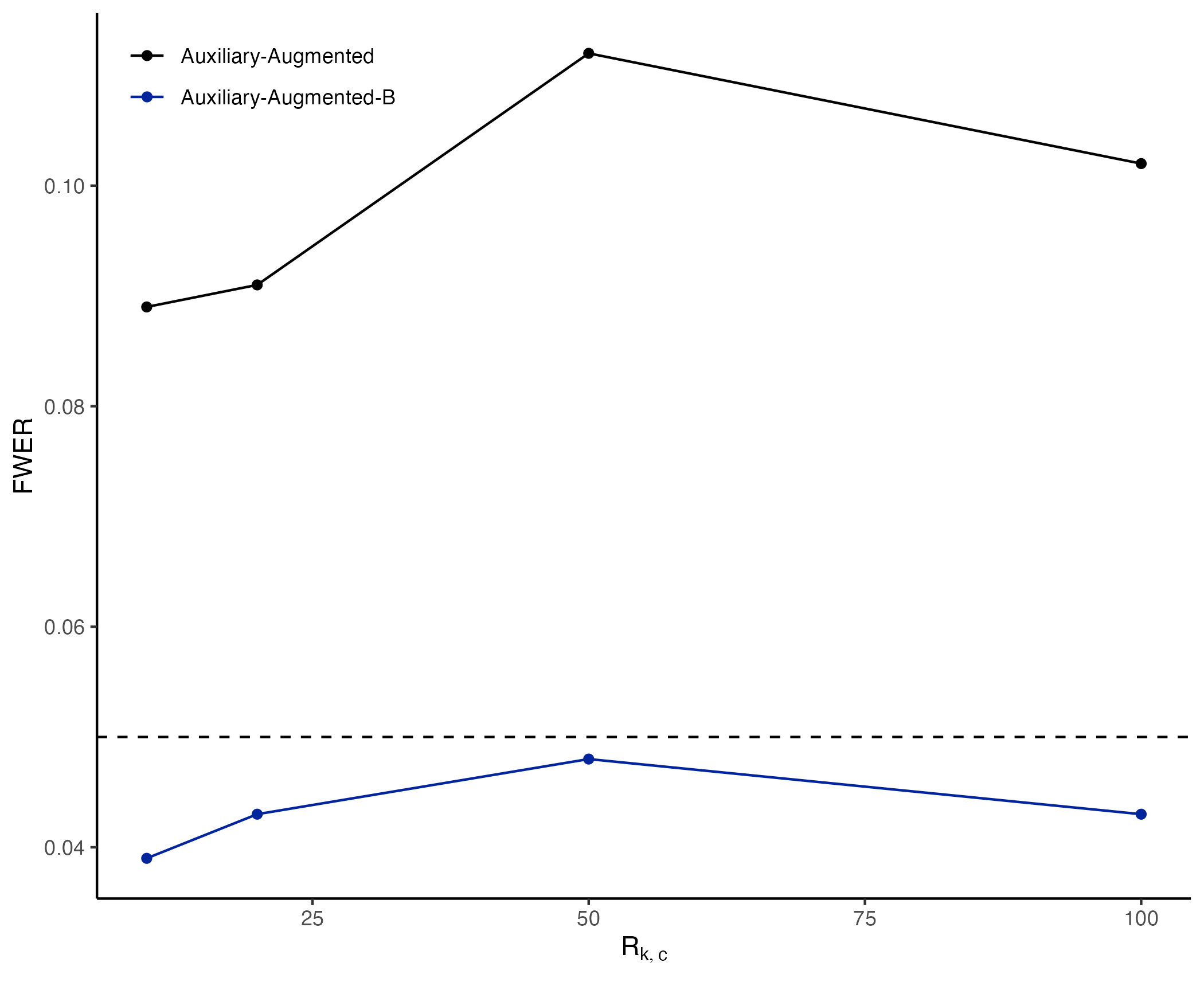}
\end{center}
\caption{We consider scenario 1 of the simulations with $K=6$ subgroups outlined in \secname~{5.1}. We increase the values of the odds ratios in the range $R_{k,c}\in \{15,20,50,100\}$.  For each $R_{k,c},$ we generate $1,000$ trials and compute the FWER for the Auxiliary-Augmented procedure and the Auxiliary-Augmented-B procedure with the optimal parameter $\beta = 11.4.$ The black dashed line is the target FWER of $0.05$.}
\label{fig:boot}
\end{figure}

\begin{figure}[h!]
	\includegraphics[width = 0.5\textwidth]{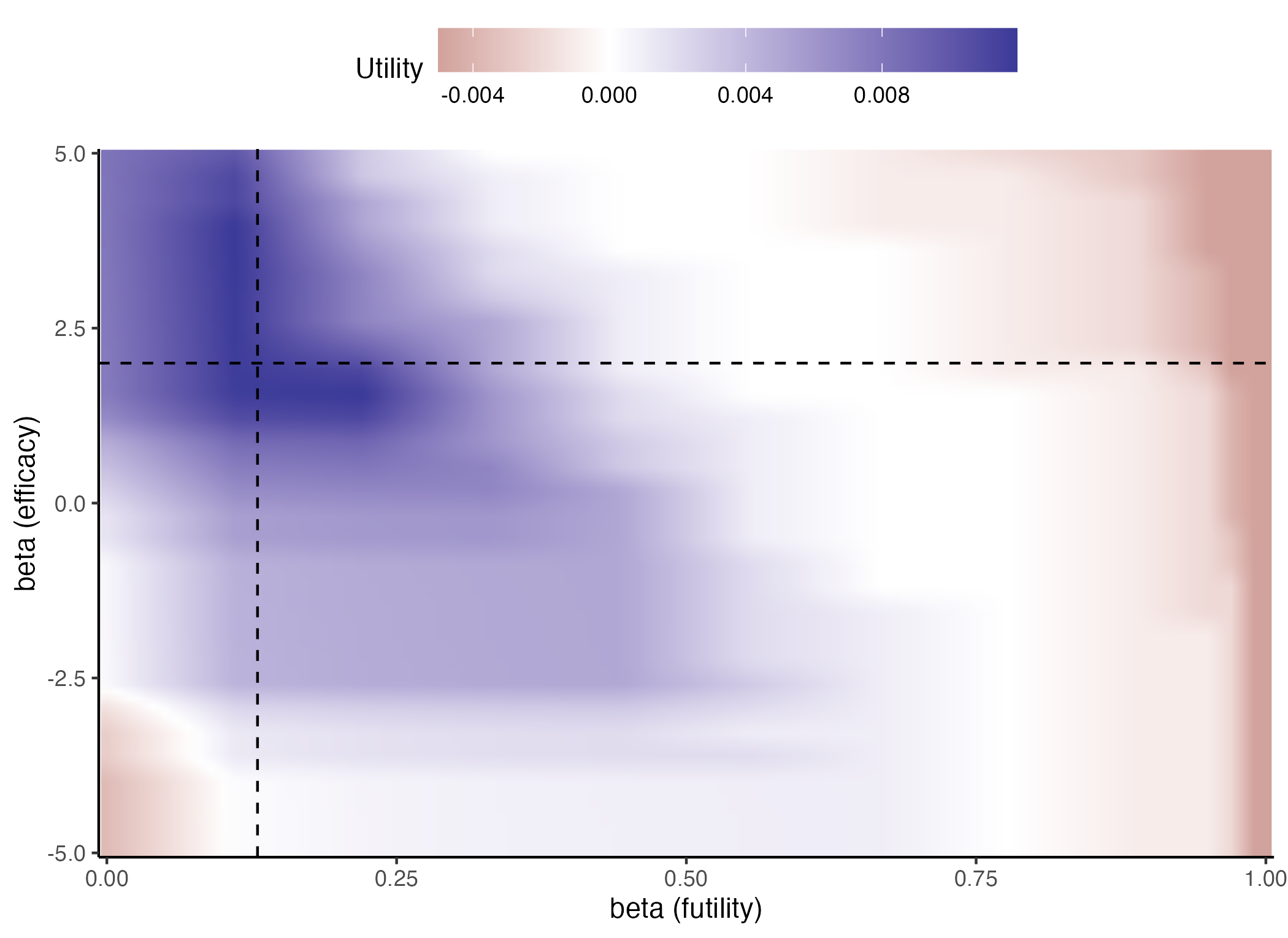}
	\includegraphics[width = 0.5\textwidth]{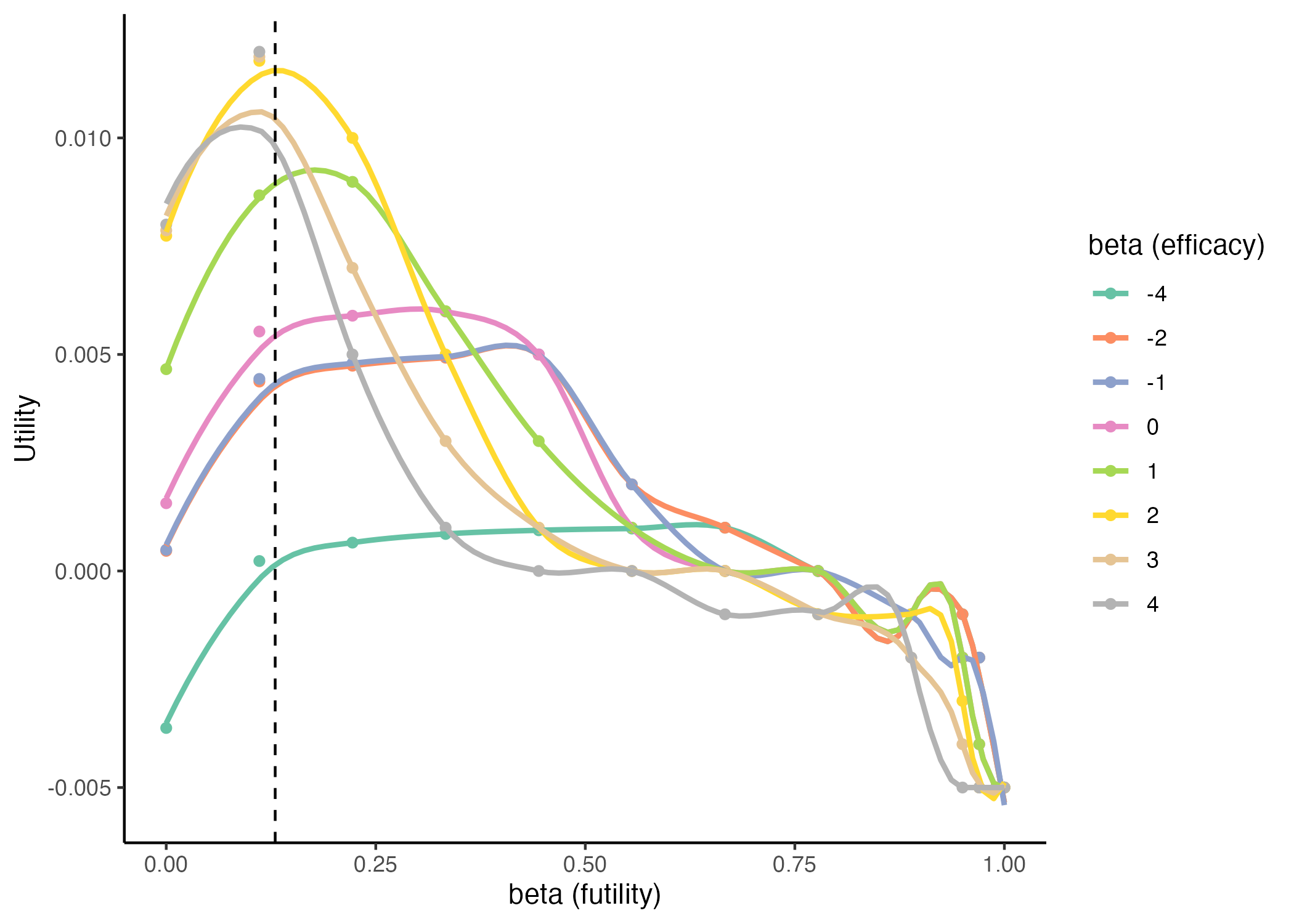}

\caption{Expected utility for the simulation study presented in \secname~{6}, with $\lambda = 0.00005,$ $\lambda^{'}_{1}  =1,$ and $\lambda^{'}_{2} = 0.5$. The left panel shows the expected utility surface, with dashed lines indicating the coordinate of the optimal solution $(\beta_e = 2, \beta_f = 0.13).$ The right panel shows the expected utility for some values of the efficacy parameter. The solid lines represent a smooth estimate of the expected utility obtained using a Local Polynomial Regression Fitting (loess) (R function \texttt{loess} with \texttt{span} = 0.4). The dashed vertical line indicates the optimal solution $\beta_f = 0.13.$} 
\label{fig:utlity_sec6}
\end{figure}

\clearpage

\section{Retrospective analysis with GBM trial data: time-to-event outcomes}\label{sec:s3}

We reimplement the retrospective analysis of the CENTRIC  trial~\citep{centric}  presented in Section 6.1 using time-to-event outcomes. The primary outcome of interest is overall survival (OS), while the auxiliary outcome is progression-free survival (PFS); both are measured from the time of enrollment in the trial.
We use the utility function in Equation (9) of the main manuscript and the decision rules described in Section 6. 

We rely on a Bayesian prior model for OS and PFS based on the work of~\citet{broglio:2009}, where the OS is decomposed in the sum of PFS and survival post-progression (SPP). In the prior model, if a treatment has a positive effect on OS, this effect is decomposed into the effect on PFS and SPP. We assume that the effect on SPP is a fraction of the effect of PFS.  % As in ~\citet{broglio:2009}, OS and PFS are assumed to be independent and exponentially distributed.
Indicating with $P_i$ the SPP for individual $i$, our prior model can be written as:  
\begin{align*}
S_i &\sim \mbox{Exp}(\theta^S_{C_i}),
\\
P_i &\sim \mbox{Exp}(\theta^P), \\
\log(1/\theta^S_{C_i}) &= \zeta^S_0 + \zeta^S_1C_i,  \\
\zeta^S_1 &\sim \xi 1\{\zeta^S = 0 \} + (1-\xi)\mathcal{N}(m_S,\sigma^2_S), \\
\zeta^S_0 &\sim \mathcal{N}(m, \sigma^2), \\
\log\{1/\theta^P_{C_i}\} &=  \xi^P_0 + c_P \zeta^S_1C_i,\\
\xi^P_0 = &\sim \mathcal{N}(m_P, \sigma^2_P), \\
c_Y &\sim \mbox{Beta}(v,o).
\end{align*}

Here $\mbox{Exp}(\theta)$ indicates an exponential distribution with mean $1/\theta.$ Note that in the formulation of \citet{broglio:2009} $c_P =0,$ that is to say, all the treatment effects are entirely mediated by PFS.
The prior for the treatment effect parameters $\zeta^S_1$ and $c_P$ follows the same logic as the model introduced in Section 5.1. Alternative models can be used to generate trials with OS and PFS. Examples include the models proposed in  \citet{renfro:2011} and \citet{chen:2020}---among others.

Values on the prior parameters used in our example are reported in the caption of Table~\ref{tab:prior_surv}, which also reports relevant characteristics of the trials generated from the prior models. 
\begin{table}[h!]
\begin{center}
\resizebox{0.9\textwidth}{!}{\begin{tabular}{r|cccccc} 
\toprule 
Characteristic                        & Mean (sd)      & Min   & 1st quantile & Median & 3rd quantile & Max    \\
\midrule 
Median OS (SOC)                       & 588.4 (389.1)  & 70.9  & 337.4        & 489.8  & 717.7        & 4976.3 \\
Median OS (Treated)                   & 600.3 (447.1)  & 42.2  & 337.0        & 487.0  & 726.2        & 7555.1 \\
Median PFS (SOC)                      & 442.7 (368.0)  & 18.3  & 211.3        & 342.3  & 547.1        & 4481.4 \\
Median PFS (Treated)                  & 457.3 (434.0)  & 14.6  & 207.0        & 341.3  & 561.1        & 7125.3 \\
Correlation between PFS and OS        & 0.9 (0.1)      & 0.0   & 0.9          & 1.0    & 1.0          & 1.0    \\
Proportion of uncersored PFS          & 0.8 (0.0)      & 0.7   & 0.8          & 0.8    & 0.8          & 0.9    \\
Proportion of uncersored OS           & 0.6 (0.0)      & 0.5   & 0.6          & 0.6    & 0.7          & 0.8    \\
Time (in days) to collect 50 OS evets & 1101.8 (218.1) & 589.6 & 933.0        & 1074.1 & 1254.8       & 1500.0 \\
\bottomrule 
\end{tabular}
}
\end{center}
\caption{Summaries of some characteristics of the prior model outlined in Section S2 of the Supplementary Material. Characteristics have been calculated  simulating $1000$ trials from the prior model, setting
$\xi =0.1,$
$m_S = 0,$
$\sigma^2_S = 0.8,$
$m=6.20,$
$\sigma^2 = 0.5,$
$m_P = 4.85,$ and
$\sigma^2_P = 0.5,$ 
$v = 6$, 
$o =1$.
}
\label{tab:prior_surv}
\end{table}
We added a censoring mechanism to our prior model to generate realistic trials. Censoring times (loss of follow-up) for each trial participant are generated from two independent exponential distributions. Similar to the actual trial, in each trial simulation, the proportion of observed PFS events is approximately $0.8$, and the proportion of observed OS events is approximately $0.6$.  We assumed an accrual rate of four patients per month. In our experience, this is a realistic accrual rate for nGBM trials. The code to generate the trials from the prior model is available at \url{https://github.com/rMassimiliano/primary\_and\_auxiliary} via the function \texttt{generate\_data\_from\_prior\_survival()}. For efficacy decisions, we used the log-rank test as in the original trial~\citep{centric}.
Other tests, such as the weighted log-rank test, could be used.
We perform one interim analysis (IA) after $50$ survival events have been collected. All the data available at the time of the IA, such as time to progression for patients still alive, are used for decision-making.

The optimizations to identify the parameters of the design is  summarized in \figurename~\ref{fig:utlity_surv}. The optimum is at $(\beta_e,\beta_f) = (4.3,0.1).$ 

\begin{figure}[h!]
\center
	\includegraphics[width = 0.5\textwidth]{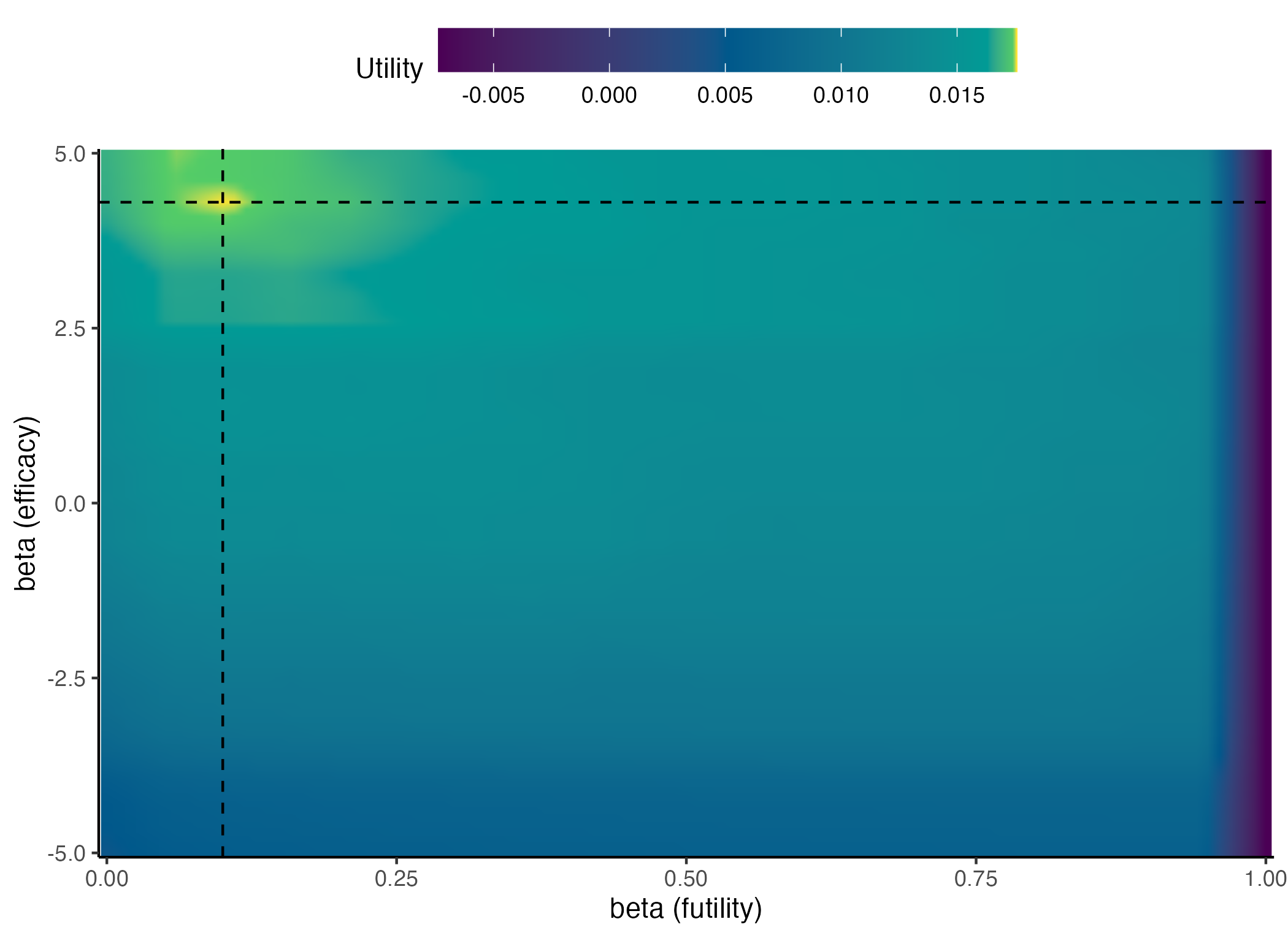}
\caption{Expected utility for the simulation study presented in \secname~{S2} of the Supplementary Materials, with $\lambda = 0.00005,$ $\lambda^{'}_{1}  =1,$ and $\lambda^{'}_{2} = 0.5$. The left panel shows the expected utility surface, with dashed lines indicating the coordinate of the optimal solution $(\beta_e = 4.3, \beta_f = 0.1).$ }
\label{fig:utlity_surv}
\end{figure}

Similarly to Section 6.1, to generate in silico trial replicates tailored to nGBM, we follow these two steps:
\begin{enumerate}
\item We sample with replacement a patient from the TMZ+RT group of the  CENTRIC dataset. The patient is randomly assigned to either the TMZ+RT or the experimental arm of our in silico trial.
\item If the patient is assigned to our in silico TMZ+RT, we include the actual OS and PFS as primary and auxiliary outcomes. If the patient is assigned to the in silico experimental arm, we use a perturbation of the actual OS and PFS to produce scenarios with positive or negative treatment effects. 
Specifically, after sampling from the control of CENTRIC,   
we multiply the observed patient's OS and PFS by $\exp\{h_y\}$ and  $\exp\{h_s\}$ respectively.
  The parameters $h_y$ and $h_s$ create in silico trials with treatment effects.  
\end{enumerate}

By adjusting the values of $h_y$ and $h_s$, we created scenarios to examine the trial design, with or without concordance of the treatment effects on primary and auxiliary outcomes. We considered three scenarios: 1. (null scenario) $h_y = h_s = 0$ there are no treatment effects on both primary and auxiliary outcome; 2. (concordant effects) $h_y = 0.55$ and $h_s = 0.5$; 3. (discordant effects) $h_y = 0$ and $h_s = 0.5,$ mimicking a trial where an effect is present for PFS but not on OS,  as for lomustine and bevacizumab~\citep{wick:2017}.

 We compare our approach (\textit{Auxiliary-Augmented} in \tablename~\ref{tab:surv_data}) with two alternative trial designs. The first, the  {\it Primary-Only} design,  uses only the primary outcome data for efficacy and futility analyses. This design uses the same group-sequential efficacy rule as the \textit{Auxiliary-Augmented} design, which is based on a $\alpha$-spending (10) with $\beta_{E} = 4.3.$  The futility decisions are based on a Bayesian model that ignores auxiliary information. Specifically, we assume that the OS is exponentially distributed, with mean $\exp\{\zeta^Y_0 +  \zeta^Y_1 C_i\}.$ We let $\zeta^Y_0 \sim \mathcal{N}(6.40, 0.5)$, and $\zeta^Y_1 \sim \xi 1\{\zeta^Y = 0 \} + (1-\xi)\mathcal{N}(0,0.8)$, with  $\xi = 0.1.$ All the remaining parameters of the design are the same of the {\it Auxiliary-Augmented}, including the futility threshold $\beta_F = 0.1$. The other design in our comparisons is the \textit{Auxiliary-Only} design. It is nearly identical to the \textit{Primary-Only} design but replaces the primary outcomes with the auxiliary outcomes for the decisions of efficacy and futility.  
\begin{table}[h!]
\center
\resizebox{\textwidth}{!}{ 
 \begin{tabular}{r|cc|cc|cc} 

    \toprule 

& \multicolumn{2}{c}{Scenario 1} & \multicolumn{2}{c}{Scenario 2} & \multicolumn{2}{c}{Scenario 3} \\  
& \multicolumn{2}{c}{null scenario} & \multicolumn{2}{c}{concordant effects} & \multicolumn{2}{c}{discordant effects} \\ 
\midrule 
Method &  Power &   Stop for futility & Power & Stop for futility & Power &   Stop for futility \\ 
Auxiliary-Augmented & 0.04 (0.03; 0.00) & 0.87 & 0.83 (0.78; 0.05) & 0.05 & 0.05 (0.04; 0.01) & 0.70 \\
Primary-Only        & 0.05 (0.04; 0.01) & 0.72 & 0.81 (0.76; 0.05) & 0.14 & 0.05 (0.04; 0.01) & 0.72 \\
Auxiliary-Only      & 0.05 (0.04; 0.01) & 0.63 & 0.88 (0.60; 0.27) & 0.03 & 0.72 (0.41; 0.31) & 0.06 \\
\bottomrule  
 \end{tabular}

}
\caption{Proportion of times the null hypothesis $H_0$ of no treatment effect 
on the primary outcome (auxiliary outcome for the \textit{Auxiliary-Only} design)
is rejected (Power), and the proportion of time the trial is stopped for futility  (``Stopping for futility'') for 5000 simulated trials using the CENTRIC data. In parenthesis,  we report the frequency of simulations in which $H_0$ is rejected at the interim and final analysis, respectively.
}
\label{tab:surv_data}
\end{table}

In scenario 1, without treatment effects, all three designs control the type-I error at the $5\%.$ The \textit{Auxiliary-Augmented} design has the largest proportion of trials stopped for futility, $87\%,$ compared to $72\%$ for the \textit{Primary-Only} design, and $63\%$ for the  \textit{Auxiliary-Only} design. In scenario 2,  with treatment effects both on primary and auxiliary outcomes,   the \textit{Auxiliary-Augmented} and \textit{Primary-Only} designs have $83\%$ and $81\%$ power respectively compared to and $88\%$ for the \textit{Auxiliary-Only} design. The \textit{Auxiliary-Augmented} design has a smaller percentage of stopping for futility (5\%) than the \textit{Primary-Only} design (14\%). This reflects that joint modeling of primary and auxiliary outcomes can effectively capture the presence of treatment effect, reducing the likelihood that a trial is erroneously stopped for futility.

In scenario 3, with a treatment effect only on the auxiliary outcomes, as expected, the \textit{Auxiliary-Only} design has an inflated type-I error $72\%$, while the  \textit{Auxiliary-Augmented} and \textit{Auxiliary-Only} designs control the type-I error at $5\%.$ Due to promising auxiliary data, the   \textit{Auxiliary-Augmented} design has a slightly higher proportion of trials stopped for futility,  $72\%$ patients compared to  $70\%$ of the \textit{Auxiliary-Only} design. Importantly, in this scenario, the \textit{Auxiliary-Augmented} prevents inflation of type-I error based on promising auxiliary data.

{%\setstretch{1.1}
\bibliographystyle{ba-modified}
\bibliography{references.bib}
}

\end{document}